\documentclass[peerreview]{IEEEtran}
\renewcommand{\baselinestretch}{1.5}
\usepackage{graphicx}
\usepackage{amssymb}
\newcommand{\qed}{\fbox{}}
\newcommand{\ve}[1]{ \mbox{\boldmath$#1$} }
\newcommand{\defeq}{\stackrel{\triangle}{=}}

\newtheorem{example}{Example}
\newtheorem{definition}{Definition}

\newtheorem{lemma}{Lemma}

\newtheorem{theorem}{Theorem}

\newcommand{\trace}{{\sf trace}}
\newcommand{\sfvec}{{\sf vec}}
\newcommand{\mod}{\ {\sf mod }\ }

\begin{document}

\title{LP Decodable Permutation Codes
based on Linearly Constrained Permutation Matrices}
\author{Tadashi Wadayama and Manabu Hagiwara
\thanks{T. Wadayama is with 
  Nagoya Institute of Technology, Nagoya City,  Aichi, 466-8555, JAPAN.
    (e-mail:wadayama@nitech.ac.jp).
    M. Hagiwara is with National Institute of
Advanced Industrial Science and Technology,
Central 2, 1-1-1 Umezono, Tsukuba City,
Ibaraki, 305-8568, JAPAN (email: hagiwara.hagiwara@aist.go.jp).
  A part of this work will be presented at International Symposium on Information Theory, 2011.
  The initial version of this work has been  included in e-preprint server arXiv since Nov. 2010 
  (identificator:{\em 	arXiv:1011.6441}).
 } }

\maketitle
\begin{abstract}
A set of linearly constrained permutation matrices are proposed for 
constructing a class of permutation codes. Making use of linear constraints imposed on 
the permutation matrices, we can formulate a minimum Euclidian distance 
decoding problem for the proposed 
class of permutation codes as a linear programming (LP) problem. 
The main feature of this class of permutation codes, called {\em LP decodable permutation codes},
is this LP decodability.
It is demonstrated that the LP decoding performance of the proposed class of permutation codes is characterized by the vertices of the
code polytope of the code. 
Two types of linear constraints are discussed; one is structured constraints and another is random constraints.
The structured constraints  such as pure involution lead to an efficient encoding algorithm.
On the other hand, the random constraints enable us to use probabilistic methods for analyzing several 
code properties such as the average cardinality and the average weight distribution.
\end{abstract}

{\bf Index Terms}: permutation codes,  linear programming,  polytope, decoding, error correction  

\section{Introduction}

The class of linear codes defined over a finite field is ubiquitously employed in digital equipments
for achieving reliable communication and storage systems.
For example, the class of codes includes practically important codes such as 
Reed-Solomon codes, BCH codes, and LDPC codes.
The linearity of codes enables us to use efficient encoding and decoding algorithms based on
their linear algebraic properties.

On the other hand,  there are some classes of nonlinear codes which are 
interesting from both theoretical and practical points of view.
The class of {\em permutation codes}  is such a class of nonlinear codes.

The origin of permutation codes dates back to the 1960s. 
Slepian \cite{Slepian} proposed a class of simple permutation codes, which is referred to as  {\em permutation modulation}, 
and efficient soft decoding algorithms for these codes.  The variant I code \cite{Slepian} is obtained by applying all the permutations  
to the initial vector
\[
(\overbrace{\mu_1, \mu_1  \ldots, \mu_{1}}^{n_1 \mbox{\scriptsize}} \overbrace{\mu_2,  \ldots, \mu_{2}}^{n_2 \mbox{\scriptsize}}
\cdots \overbrace{\mu_k, \mu_k  \ldots, \mu_{k}}^{n_k \mbox{\scriptsize}}),
\]
where $\mu_i$ is a real value and  
$n = n_1+\cdots + n_k$. 
This research has been extended and investigated by a number of researchers.
Biglieri and Elia \cite{Biglieri}, Karlof \cite{Karlof}, Ingemarsson \cite{Ingemarsson}
studied optimization of the initial vector of the permutation modulation. Berger et al. \cite{Berger}
discussed applications of permutation codes to source coding problems.

There is another thread of researches on a class of permutation codes of length $n$ 
whose codeword contains exactly $n$-distinct symbols; 
i.e., any codeword can be 
obtained by applying a permutation to an initial vector, e.g., $(0,1,\ldots, n-1)$.

Some fundamental properties of such permutation codes 
were discussed in Blake et al. \cite{Blake1979}, and Frankl and Deza \cite{Frankl1977}.
Vinck \cite{Han2000} \cite{Han2000-2} proposed applications of permutation codes 
for power-line communication and this triggered subsequent works on 
permutation codes. 
Wadayama and Vinck \cite{Wadayama2001} presented a multi-level construction of
permutation codes with large minimum Hamming distance. 
A number of constructions for permutation codes have been developed,
including the construction given in \cite{Colbourn} \cite{Klove2002}.
Especially,  the idea of a distance-preserving map due to Vinck and Ferreira \cite{Ferreira}
had influence on the study of permutation codes such as subsequent works by
Chang et al. \cite{Chang2003} \cite{Chang2005}.

Recently, rank modulation codes for flash memory proposed by Jiang et al. \cite{Jiang2008-1} \cite{Jiang2008-2} 
generated  renewed interest in permutation codes. 
For example, for flash memory coding, Kl\o ve et al. gave a new construction for permutation codes 
based on Chebyshev Distance \cite{Klove2010}, which is an appropriate distance measure for flash memory coding.
Barg  and Mazumdar \cite{Barg} also studied some fundamental bounds on permutation codes in terms of the Kendall tau distance. 

In order to employ a permutation code in a practical application,  efficient encoding and 
soft-decoding algorithms are crucial to achieve reliable communication over noisy channels, 
such as an AWGN channel. Nonlinearity of permutation codes prevents the use of 
conventional encoding and decoding techniques based on linear algebraic properties. 
Although much works on permutation codes have been conducted,  an aspect of efficient soft-decoding 
has not been intensively discussed so far.
Therefore, there is still room for further researches on permutation codes with efficient encoding and 
soft-decoding algorithms.

In this paper, a new class of permutation codes called {\em LP decodable permutation codes} 
is introduced.  An LP decodable permutation code is 
obtained by applying permutation matrices satisfying certain linear constraints 
to an $n$-dimensional real initial vector. 

It is well known that permutation matrices 
are vertices of the Birkhoff polytope \cite{Birkhoff}, which is the set of doubly stochastic matrices.
Thus, a set of linearly constrained permutation matrices can be expressed by a set of 
linear equalities and linear inequalities.
This property leads to the main feature of this class of permutation codes: {\em LP-decodable property}.
For this class of codes, a decoding  problem can be formulated as a linear programming (LP) problem.  
This means that we can exploit efficient LP
solvers based on simplex methods or interior point methods to decode LP decodable permutation codes. 

Furthermore, for a combination of this class of codes and its LP decoding, 
the maximum likelihood (ML) certificate property can be proved 
as in the case of the LP decoding for LDPC codes \cite{Feldman}.
This is due to the fact that the LP problem given in this paper is a relaxed problem of 
an ML decoding problem.

In general, a fundamental polytope \cite{KV}  \cite{Feldman} 
used for LP decoding of LDPC codes contains a number of fractional 
vertices, which are a major source of sub-optimality of LP decoding. 
The constraints corresponding to an LDPC matrix 
are defined based on $\Bbb F_2$-arithmetics. On the other hand, an LP decoder works on the real number field. This domain mismatch produces  many undesirable fractional vertices on the fundamental polytope. One motivation of the present study is to establish a coding scheme without this mismatch. In other words, the LP decodable permutation codes are defined on the real number field and are decoded using an LP solver working on the real number field.

The organization of the paper is as follows.
Section \ref{preliminaries}  introduces some definitions and notation required for discussion. 
Section \ref{LCP} gives the definition of the LP decodable permutation codes and its decoding algorithm.
Section \ref{decodinganalysis} provides analysis for decoding performance of LP decoding and ML decoding. 
Section \ref{structured} presents some classes of permutation codes which are easy to encode.
Section \ref{random} offers probabilistic analysis on the cardinality and weight distribution of random LP decodable permutation codes.
Section \ref{conclusion} gives a concluding summary.

\section{Preliminaries}

\label{preliminaries}
\subsection{Notation and definition}
In this paper, matrices are represented by capital letters
and a vector is assumed to be a column vector.
Let $X$ be an $n \times n$ real matrix.
The notation $X \ge 0$ means that every element in $X$ is non-negative.
The notation $\sfvec(X)$ represents  a vectorization of $X$  given by
\[
\sfvec(X) \defeq
\left(
X_{1,1}\  \cdots X_{1,n}  \ X_{2,1}  \  \cdots X_{2,n}, X_{3,1} \cdots   X_{n,n}
\right)^T.
\]

The vector $\ve 1$ is the all-one vector whose length is determined by the context.
The norm $||\cdot ||$ denotes the Euclidean norm given by $|| x || \defeq (x^T  x)^{1/2}$.
The trace function $\trace(X)$ returns the sum of the diagonal elements of $X$.
The sets $\Bbb R, \Bbb Z$ are the sets of real numbers and integers, respectively.
The set $[\alpha,\beta]$ denotes the set of consecutive integers from $\alpha \in \Bbb Z$ to $\beta \in \Bbb Z$.

The symbol $\unlhd$ means 
\[
\left(
\begin{array}{c}
a_1 \\
\vdots \\
a_m
\end{array}
\right)
\unlhd
\left(
\begin{array}{c}
b_1 \\
\vdots \\
b_m
\end{array}
\right) \Leftrightarrow
\forall i \in [1,m],  a_i \unlhd_i b_i,
\]
where $\unlhd_i$ is either $=$ or $\le$.
For simplicity, the notation $\unlhd = (\unlhd_1, \unlhd_2, \ldots, \unlhd_m)^T$ is used to define $\unlhd$ 
(e.g.,  $\unlhd = (\le, =, \le)^T$).

The next definition gives a class of matrices of crucial importance in this paper.
\begin{definition}[Permutation matrix]
An $n \times n$ binary real matrix $X \defeq (X_{i,j})_{i,j \in [1,n] } \in \{0,1\}^{n \times n} $
is called a {\em permutation matrix} if and only if 
\begin{equation} \label{permcond}
 \forall i, j \in [1, n], \sum_{j' \in [1,n]} X_{i, j'} = 1, \sum_{i' \in [1,n]} X_{i', j} = 1.
 \end{equation}
 \hfill\fbox{}
\end{definition}
The set of $n \times n$ permutation matrices is denoted by $\Pi_n$.  The cardinality of $\Pi_n$ is $n!$.

Removing the binary constraint from the definition of the permutation matrices, we have
the definition of doubly stochastic matrices.
\begin{definition}[Doubly stochastic matrix]
An $n \times n$ non-negative real matrix $X \defeq (X_{i,j})_{i,j \in [1,n] }$
is called a {\em doubly stochastic matrix } 
if and only if (\ref{permcond}) holds.
 \hfill\fbox{}
\end{definition}

The following theorem for a double stochastic matrix implies that the set of 
doubly stochastic matrices is a convex polytope.
\begin{theorem}[Birkhoff-von Neumann theorem \cite{Birkhoff} \cite{vonNeumann} ]
Every doubly stochastic matrix is a convex combination of permutation matrices.
\end{theorem}

The set of $n \times n$ doubly stochastic matrices
is a  polytope called the {\em Birkhoff polytope} $B_n$ \cite{Birkhoff}, which is also known as perfect matching polytope.
The Birkhoff polytope is a $(n-1)^2$-dimensional convex polytope with  $n!$-vertices and $n^2$-facets \cite{Ziegler}.
The Birkhoff-von Neumann theorem implies that any vertex (i.e., extreme point) of the Birkhoff polytope is a permutation matrix
and vice versa.

\subsection{LP decoding for permutation vectors}

Assume that 
$
s \in \Bbb R^n
$,
called the {\em initial vector}, is given\footnote{The elements in $s$ are not necessarily distinct each other.}.
The set of images of $s$ by left action of $X \in \Pi_n$ is called the {\em permutation vectors} of
$s$, which is given by
\begin{equation}
\Lambda(s) \defeq  \{X s \mid  X \in \Pi_n\}.
\end{equation}
For example, if $s = (0,1,2)^T$, then $\Lambda(s)$ is given by
\[
\Lambda(s) = \{(0,1,2), (0,2,1),(1,0,2),(1,2,0),  (2,0,1) (2,1,0) \}.
\]

We here consider a situation such that a vector of $\Lambda(s)$ is transmitted to
a receiver over an AWGN channel. In such a case, 
it is desirable to use an ML decoding algorithm to estimate the transmitted vector.
The ML decoding rule can be describe as 
\begin{equation}
\hat {x} = \arg \min_{x \in \Lambda(s)} ||y - x||^2,
\end{equation}
where $y$ is a received word.

The next theorem states that the ML decoding for $\Lambda(s)$ can be formulated as
the following LP problem.

\begin{theorem}[LP decoding and ML certificate property]
Assume that a vector in $\Lambda(s)$ is transmitted over an AWGN channel and 
that $y \in \Bbb R^n$ is received on the receiver side.
We also suppose that 
$\hat {x} =\arg \min_{x \in \Lambda(s)} ||y - x||^2$ is uniquely determined from $y$. 
Let $X^*$ be the solution of the following LP problem:
\begin{eqnarray} \nonumber
\mbox{maximize }  \trace(C^T X) \\ \nonumber
\mbox{subject to } \\ \nonumber
X &\in& \Bbb R^{n \times n} \\ \nonumber
X \ve {1} &=& \ve {1} \\ \nonumber
\ve{1}^T X  &=& \ve {1}^T \\ \label{LP}
X  &\ge& 0, 
\end{eqnarray}
where $C \defeq y s^T$.
If $X^*$ is integral, $\hat {x}=X^* s$ holds.
\end{theorem}
\begin{proof}
The linear constraints in the above LP problem implies 
that $X$ is constrained to be a doubly stochastic matrix.

On the other hand, the ML decoding rule can be recast as follows:
\begin{eqnarray} \nonumber
\hat {x}  
&=& \arg \min_{x \in \Lambda(s)} ||y - x||^2  \\ \nonumber
&=& (\arg \min_{X \in \Pi_n} || y - X s||^2)  s \\ \nonumber
&=& (\arg \min_{X \in \Pi_n} (|| y ||^2 - 2 y^T (X  s) + ||X s||^2)) s \\ \nonumber
&=& (\arg \max_{X \in \Pi_n}  y^T X s) s
= (\arg \max_{X \in \Pi_n}  \trace(C^T X))  s,
\end{eqnarray}
where $C = y s^T$. 
Note that 
\begin{equation}
\trace(C^T X) = \sum_{i=1}^n \sum_{j=1}^n C_{i,j} X_{i,j}.
\end{equation}
Since the vertices of the Birkhoff polytope is a permutation matrix,
the ML decoding can be formulated as an integer LP (ILP) problem:
\begin{eqnarray} \nonumber
&&\mbox{maximize }  \trace(C^T X) \\ \nonumber
&&\mbox{subject to } 
X \in B_n,\quad X  \mbox{\   is an integral matrix}.
\end{eqnarray}
By removing the integral constraint ($X$ is an integral matrix), we obtain the LP problem (\ref{LP}).
If the solution of this LP problem is integral, it must coincide with the solution of the above 
ILP problem.
\end{proof}

As we have seen,   the feasible set of the above LP  problem is the Birkhoff polytope.
Thus, an output of the above LP is highly likely integral.

The following example illustrates an LP decoding procedure.
\begin{example}
Let  $s \defeq (0,1)^T$. In this case, 
the set of permutation vectors becomes
$
\Lambda(s) = \{(0,1)^T, (1,0)^T\}.
$
Assume that $y = (0.9, 0.2)^T$ is received. In this case, 
\[
C = y s^T = 
\left(
\begin{array}{c}
0.9 \\
0.2 \\
\end{array}
\right) (0\  1)
= 
\left(
\begin{array}{cc}
0 & 0.9 \\
0 & 0.2 \\
\end{array}
\right)
\]
is obtained. By letting 
\[
X = 
\left(
\begin{array}{cc}
X_{1,1} & X_{1,2} \\
X_{2,1} & X_{2,2} \\
\end{array}
\right),
\]
we have the objective function 
\[
\trace
\left(
\left(
\begin{array}{cc}
0 & 0 \\
0.9 & 0.2 \\
\end{array}
\right)
\left(
\begin{array}{cc}
X_{1,1} & X_{1,2} \\
X_{2,1} & X_{2,2} \\
\end{array}
\right)
\right)= 0.9 X_{1,2} + 0.2 X_{2,2}.
\]
As a result, the LP decoding problem is given by
\begin{eqnarray} \nonumber
&&\mbox{maximize } 0.9 X_{1,2} + 0.2 X_{2,2}
\mbox{ subject to } \\ \nonumber
&& X_{1,1}+ X_{1,2} = 1, \quad X_{2,1}+ X_{2,2}=1,  \\ \nonumber
&& X_{1,1} + X_{2,1} = 1, \quad X_{1,2} + X_{2,2}=1 \\ \nonumber
&& X_{1,1}, X_{1,2}, X_{2,1}, X_{2,2} \ge 0.
\end{eqnarray}
The solution  of the problem is 
\[
X^* = 
\left(
\begin{array}{cc}
0 & 1 \\
1 & 0 \\
\end{array}
\right),
\]
and then we have the estimated word $X^* s = (1,0)^T$. \hfill\qed
\end{example}

\section{Linearly constrained permutation matrices and LP decodable permutation codes}
\label{LCP}
It is natural to consider an extension of 
the LP decoding presented in the previous section.
Additional linear constraints imposed on $\Pi_n$ produce a restricted set of $\Lambda(s)$.
A decoding problem of such a set can be formulated as an LP problem, as in the case of the ML decoding of $\Lambda(s)$.

\subsection{Definitions}

The next definition for linearly constrained permutations gives an LP-decodable subset of $\Lambda(s)$.
\begin{definition}[linearly constrained permutation matrix]
\label{lcpm}
Let $m, n$ be positive integers.
Assume that 
$
A \in \Bbb Z^{m \times n^2},\  b \in \Bbb Z^m
$
and  $\unlhd \in \{=, \le\}^m$ are given.
A set of {\em linearly constrained permutation matrices} is defined by
\begin{equation}
\Pi(A, b, \unlhd ) \defeq \{X \in \Pi_n \mid   A \ \mbox{\sfvec}(X) \unlhd b \}.
\end{equation}
\hfill\qed 
\end{definition}
Note that $A \ \mbox{\sfvec}(X) \unlhd b$ formally represents additional $m$ equalities and inequalities.
These additional constraints provide a restriction on permutation matrices.

From the linearly constrained permutation matrices, LP decodable permutation codes 
are naturally defined as follows.
\begin{definition}[LP decodable permutation code]
Assume the same set up as in  Definition \ref{lcpm}. Suppose also that 
$s \in \Bbb R^n$ is given. The set of vectors $\Lambda(A, b, \unlhd, s)$ given by
\begin{equation}
\Lambda(A, b, \unlhd, s) \defeq \{X s \in \Bbb R^n \mid X \in \Pi(A, b, \unlhd) \}
\end{equation}
is called an LP decodable permutation code.  \hfill \fbox{}
\end{definition}

If  
$
\Rightarrow X^{(1)} s \ne X^{(2)} s
$
holds for any $X^{(1)}, X^{(2)} (X^{(1)} \ne X^{(2)})\in  \Pi(A, b, \unlhd)$,
then an LP decodable permutation code is said to be {\em non-singlar}.
Namely, there is one-to-one correspondence between permutation matrices in $\Pi(A, b, \unlhd)$
and codewords of $\Lambda(A, b, \unlhd, s)$ if a code is non-singular. 
Note that a code may become singular if identical symbols exist  in $s$.

The next example shows a case where an additional linear constraint imposes
a restriction on permutation matrices.
\begin{example}
Consider the set of linearly constrained permutation matrices 
which consists of $4 \times 4$ permutation matrices satisfying 
the linear constraint 
$
\trace(X) = 0.
$
The constraint implies that the diagonal elements of the permutation matrices are constrained to be zero.
This means that such permutation matrices correspond to permutations without fixed points, which are
called {\em derangements}.
For $n=4$, there are 9-derangement permutation matrices as follows:
\begin{eqnarray} \nonumber
\left(
\begin{array}{c}
0100 \\
1000 \\
0001 \\
0010 \\
\end{array}
\right)
\left(
\begin{array}{c}
0100 \\
0010 \\
0001 \\
1000 \\
\end{array}
\right)
\left(
\begin{array}{c}
0100 \\
0001 \\
1000 \\
0010 \\
\end{array}
\right) \\ \nonumber
\left(
\begin{array}{c}
0010 \\
1000 \\
0001 \\
0100 \\
\end{array}
\right)
\left(
\begin{array}{c}
0010 \\
0001 \\
1000 \\
0100 \\
\end{array}
\right)
\left(
\begin{array}{c}
0010 \\
0001 \\
0100 \\
1000 \\
\end{array}
\right) \\ \nonumber
\left(
\begin{array}{c}
0001 \\
1000 \\
0100 \\
0010 \\
\end{array}
\right)
\left(
\begin{array}{c}
0001 \\
0010 \\
1000 \\
0100 \\
\end{array}
\right)
\left(
\begin{array}{c}
0001 \\
0010 \\
0100 \\
1000 \\
\end{array}
\right).
\end{eqnarray}

In this case, the triple $(A,b, \unlhd)$ is defined by 
\begin{equation}
A = \sfvec(I), \quad b = 0,\quad \unlhd = (=),
\end{equation}
where $I$ is the $4 \times 4$ identity matrix.
Multiplying these matrices to the initial vector $s = (0,1,2,3)^T$ from left, we immediately obtain  the members of 
$\Lambda({A}, b, \unlhd, (0,1,2,3)^T ) $:
\begin{equation}
\begin{array}{ccc}
(1,  0,  3,  2)^T, &(1,  2,  3,  0)^T, &(1,  3, 0,  2)^T, \\
(2,  0,  3,  1)^T, &(2,  3,  0,  1)^T, &(2,  3,  1,  0)^T,  \\
(3,  0,  1,  2)^T, &(3,  2,  0,  1)^T, &(3,  2,  1,  0)^T.
\end{array}
\end{equation}
This code is thus non-singular. If the initial vector is 
\[
s=(0,0,0,0)^T,
\]
then the resulting code has the only codeword $(0,0,0,0)$. In this case, the code becomes singular.
\hfill\qed
\end{example}

\subsection{LP decoding for LP decodable permutation codes}

The LP decoding of $\Lambda({ A}, b, \unlhd,s)$ is a natural extension of the LP decoding for $\Lambda(s)$.
Assume that a vector in $\Lambda({ A}, b, \unlhd,s)$ is transmitted over an AWGN channel 
and $y \in \Bbb R^n$ is given. The procedure for the LP decoding of $\Lambda({ A}, b,\unlhd, s)$ is 
given as follows. 

\vspace{1cm}
\noindent {\bf LP decoding for an LP decodable permutation code}
\hrule
\begin{enumerate}
\item Solve the following LP problem and let $X^*$ be the solution.
\begin{eqnarray} \nonumber
\mbox{maximize }  \trace(C^T X) \\ \nonumber
\mbox{subject to} \\ \nonumber
X &\in& \Bbb R^{n \times n}, \\ \nonumber
X  &\ge& 0,  \\ \nonumber
X \ve{1} &=& \ve{1},  \\ \nonumber
\ve{1}^T X  &=& \ve{1}^T, \\  
A \ \mbox\sfvec(X) &\unlhd& b,  \label{LPL}
\end{eqnarray}
where $C = y s^T$.
\item Output $X^* s$ if $X^*$ is integral. Otherwise, declare decoding failure.
\end{enumerate}
\hrule 

\subsection{Remarks}
Several remarks should be made regarding the LP decoding for $\Lambda({A}, b, \unlhd,s)$.

The feasible set of (\ref{LPL}) is a subset of the feasible set of (\ref{LP}).
All the matrices in $\Pi(A, b, \unlhd)$ are feasible and permutation matrices which do
not belong to $\Pi(A, b,\unlhd)$ are infeasible.
This implies that all the integral points of the feasible set (\ref{LPL}) coincide with $\Pi({A },  b,\unlhd)$.

The LP problem (\ref{LPL}) is a relaxed problem of the ML decoding problem over AWGN channels:
\begin{equation} \label{mdrule}
\mbox{minimize } || y - x||^2 \mbox{ subject to } x \in \Lambda(A, b,\unlhd, s).
\end{equation}
This can be easily shown, as in the case (\ref{LP}).
As a consequence of the above properties on integral points and on the relaxation, 
it can be concluded that the LP decoding for $\Lambda(A, b, \unlhd, s)$ has the ML-certificate property as well. Namely,
if the output of LP decoding is not decoding failure (i.e., $X^*$ is integral), the output is exactly 
the same as the solution of the minimum distance decoding problem (\ref{mdrule}).
Note that the LP decoding presented above becomes the ML decoding 
if the code polytope is integral.

The feasible set of the LP problem (\ref{LPL}) is the intersection of the Birkhoff polytope and 
a (possibly unbounded) convex set 
defined by the additional constraints.
The intersection becomes a polytope which is called  a {\em code polytope}.
The decoding performance of LP decoding is closely related to the code polytope given by the following definition.
\begin{definition}[Code polytope]
The polytope $\mathcal{P}(A, b, \unlhd )$ defined by
\begin{equation}
\mathcal{P}(A, b, \unlhd ) \defeq B_n \cap \{X \in \Bbb R^{n \times n} \mid  A \ \mbox{\sfvec}(X) \unlhd b \}
\end{equation}
is called the code polytope for $\Pi(A, b, \unlhd )$, where $B_n$ is 
the Birkhoff polytope corresponding to $\Pi_n$. 
\hfill\fbox{}
\end{definition}

Figure \ref{codepolytope} illustrates a code polytope.
It should be remarked that the set of integral vertices of the code polytope coincides with
$\Pi(A, b, \unlhd)$.
Due to additional linear constraints $A \ \mbox\sfvec(X) \unlhd b$, 
a code polytope may have some fractional vertices, which contain components of fractional number.
\begin{figure}[htbp]
\begin{center}
\includegraphics[scale=0.5]{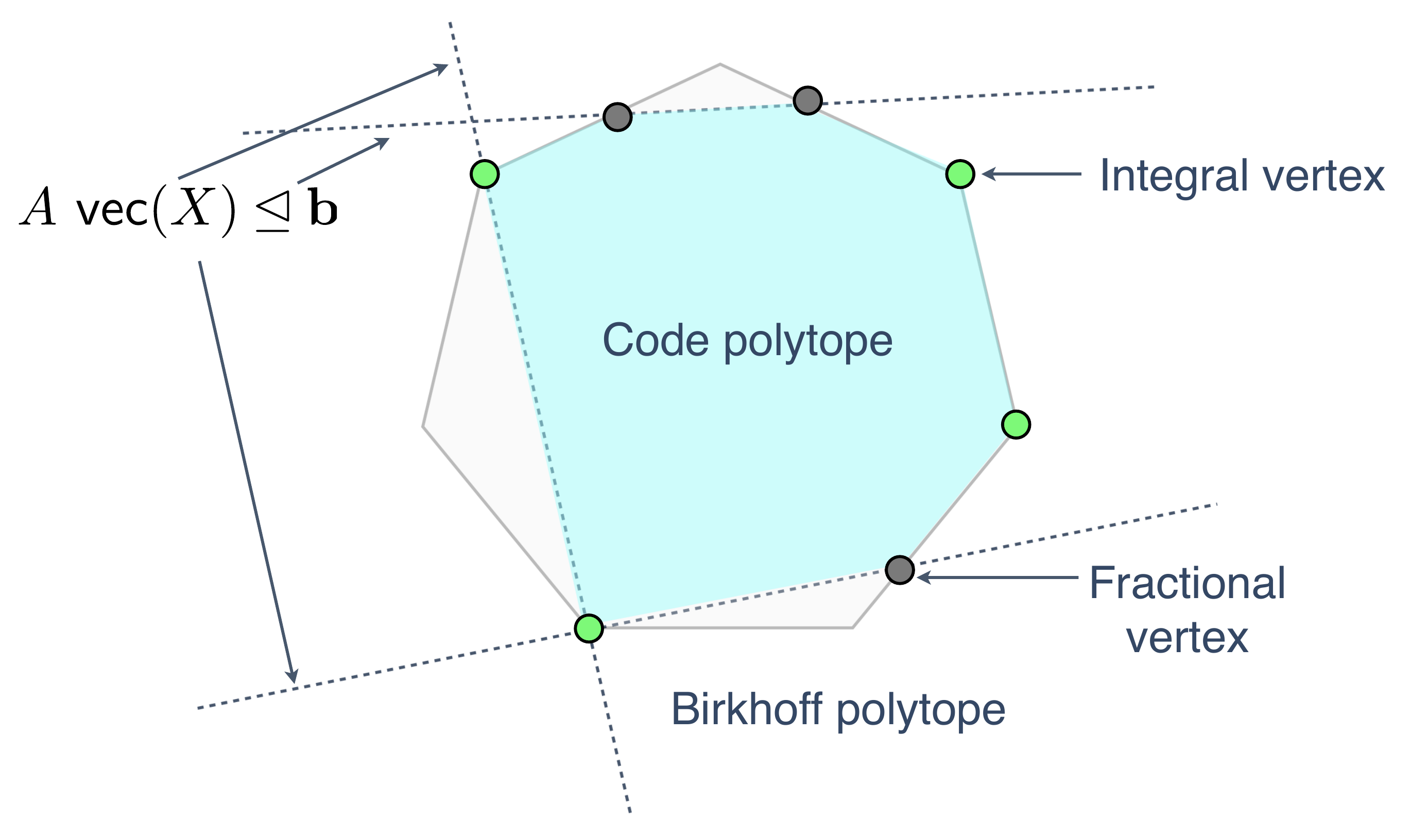}
\caption{Code polytope $\mathcal{P}(A, b, \unlhd )$}
  \label{codepolytope}
\end{center}
\end{figure}

In an LP decoding process,  these fractional vertices become  possible candidates of 
an LP solution. Thus, these fractional vertices can be considered as
{\em pseudo permutation matrices} which degrade the decoding performance of the LP decoding.

\section{Analysis for decoding performance of LP decoding and ML decoding}
\label{decodinganalysis}

In this section, upper bounds on decoding error probability for LP decoding and ML decoding 
are presented.

\subsection{Upper bound on LP decoding error probability}

An advantage of the LP formulation of a decoding algorithm is its simplicity for detailed decoding performance  analysis.
The geometrical properties of a code polytope is closely related to its decoding performance of the LP decoding.
We can evaluate the block error probability of the proposed scheme with reasonable accuracy 
if we have enough information on the set of vertices of a code polytope.
The bound presented in this section has close relationship to the pseudo codeword analysis on 
LDPC codes \cite{Forney}.

In this section, a set of parameters $A, b, \unlhd, s$ are assumed to be given.
Let $V$ be the set of vertices of the code polytope ${\cal P}(A, b, \unlhd,s)$.
In general, $V$ contains fractional vertices.

The next lemma gives bridge between a code polytope and corresponding decoding error probability.
\begin{lemma}[Upper bound on block error rate for LPD] \label{upperlemma}
Assume that a codeword $X s$ is transmitted to a receiver via an AWGN channel, where 
$X \in \Pi(A, b, \unlhd)$.
The additive white Gaussian noise with mean 0 and variance $\sigma^2$ is assumed. 
The receiver uses the LP decoding algorithm presented in the previous section.
In this case, the block error probability $P_{LP}(X)$ is upper bounded by
\begin{equation}
P_{LP}(X) \le \sum_{\tilde X \in V \backslash \{X\} } 
Q\left(\frac{||X s||^2 - (\tilde X s)^T X s }{\sigma ||\tilde X s - X s  ||} \right),
\end{equation}
where the Q-function is the tail probability of the normal Gaussian distribution, which is given by
\begin{equation} 
Q(x) \defeq \int_x^\infty \frac{1}{\sqrt{2 \pi}} \exp \left(- \frac{t^2}{2} \right) dt.
\end{equation}
\vspace{1mm}
\end{lemma}

\begin{proof}
Let $y = X s + z$, where $z$ is an additive white Gaussian noise term.
We first consider the pairwise block error probability $P_e(X, \tilde X)$ between $X$ and $\tilde X \in \Pi(A, b, \unlhd)$, 
which is given by
\begin{equation}
P_{e}(X, \tilde X) \defeq Prob[y^T \tilde X s \ge y^T X s].
\end{equation}
Namely, $P_{e}(X, \tilde X)$ is the probability such that $\tilde X s$ is more likely than $X s$ for a given $y$
under the assumption that only $\tilde X$ and $X$ are allowable permutation matrices.

The difference $y^T \tilde X s -  y^T X s$ can be transformed into 
\begin{eqnarray} \nonumber
y^T \tilde X s -  y^T X s 
&=& (X s + z)^T (\tilde X s - X s) \\ \nonumber
&=& (\tilde X s - X s)^T z +  (\tilde X s - X s)^T X s  \\ \nonumber
&=& (\tilde X s - X s)^T z  \\ 
&-& ( ||X s||^2 -  (\tilde X s)^T X s ).
\end{eqnarray}
We thus have
\begin{equation}
Prob[y^T \tilde X s \ge y^T X s] = Prob[a^T  z \ge b ],
\end{equation}
where $ a \in \Bbb R^n$ and $b \in \Bbb R$ are given by
\begin{eqnarray}
 a &\defeq& \tilde X  s - X  s, \\
b    &\defeq&  ||X  s||^2  - (\tilde X  s)^T X  s.
\end{eqnarray}

The left-hand side of $ a^T  z \ge  b$ is a linear combination of Gaussian noises.
The mean of $ a^T  z$ is zero  and the variance is given by
\begin{equation}
Var[ a^T  z ] = \sigma^2 || a||^2.
\end{equation}
The probability such that the Gaussian random variable $ a^T  z$ takes a value larger than or equal to $b$ 
can be expressed as 
\begin{eqnarray} \nonumber
P_{e}(X, \tilde X) &=& Prob[ a^T  z \ge b ] \\
&=& Q\left(\frac{b}{\sigma || a||} \right).
\end{eqnarray}
Combining the union bound and this pairwise error probability, we immediately obtain the claim of this lemma.
\end{proof}

The upper bound on decoding error probability in Lemma \ref{upperlemma} naturally leads to 
a pseudo distance measure on $\Bbb R^{n \times n}$. 
\begin{definition}[Pseudo distance]
The function 
\begin{equation}
D_{ s}(X, \tilde X) \defeq \frac{||X  s||^2 - (\tilde X  s)^T X  s }{||\tilde X  s - X s  ||}
\end{equation}
is called the {\em pseudo distance} where $X, \tilde X \in \Bbb R^{n \times n}$ are doubly stochastic matrices.
\hfill\qed
\end{definition}
Note that $D_{ s}(\cdot, \cdot)$ is not a distance function since it does not satisfy the axioms of distance.
In terms of decoding error probability, geometry of the vertices of a code polytope should be established 
based on this pseudo distance.

For example, in high SNR regime, the {\em minimum pseudo distance}
\begin{equation} 
\Delta_{ s} \defeq \min_{X\in  \Pi(A, b, \unlhd ), \tilde X \in V, \tilde X \ne X }D_{ s}(X, \tilde X)
\end{equation}
is expected to be highly influential to the decoding error probability.

\begin{example}
\label{transpos-ex}
Suppose the linear constraint $\trace(X) = 1$ where $n=3$. 
In this case, the code polytope 
has the following 5-vertices:
\begin{eqnarray} \nonumber
M^{(1)} &\defeq&\left(
\begin{array}{ccc}
1 &  0 &  0 \\
0 & 0 & 1 \\
0 & 1 & 0 \\
\end{array}
\right),\ 
M^{(2)}\defeq\left(
\begin{array}{ccc}
0 &  1 &  0 \\
1 & 0 & 0 \\
0 & 0 & 1 \\
\end{array}
\right),  \\ \nonumber
M^{(3)}&\defeq&\left(
\begin{array}{ccc}
0 &  0 &  1 \\
0 & 1 & 0 \\
1 & 0 & 0 \\
\end{array}
\right),\  
M^{(4)}\defeq\left(
\begin{array}{ccc}
1/3 &  0 &  2/3 \\
2/3 & 1/3 & 0 \\
0 & 2/3 & 1/3 \\
\end{array}
\right),  \\
M^{(5)}&\defeq&
\left(
\begin{array}{ccc}
1/3 &  2/3 &  0 \\
0 & 1/3 & 2/3 \\
2/3 & 0 & 1/3 \\
\end{array}
\right).
\end{eqnarray}
In this case, the set of vertices consists of 3-integral vertices and 2-fractional vertices. 
Let $ s = (0, 1,2)^T$. The pseudo distance distribution form $M^{(1)}$ is given by
\begin{eqnarray} \nonumber
D_{ s} (M^{(1)}, M^{(2)}) &=& 1.388730 \\ \nonumber
D_{ s} (M^{(1)}, M^{(3)}) &=& 1.224745 \\ \nonumber
D_{ s} (M^{(1)}, M^{(4)}) &=& 1.224745 \\ \nonumber
D_{ s} (M^{(1)}, M^{(5)}) &=& 1.224745. 
\end{eqnarray}
\hfill\fbox{}
\end{example}

\subsection{Upper bound on ML decoding error probability}

Assume the same setting as in the previous subsection.
In the case of ML decoding, we can neglect the effect of fractional vertices. 
Therefore,  we obtain an upper bound on the ML block error probability
\begin{eqnarray} \nonumber
P_{ML}(X) &\le&
 \sum_{\tilde X \in \Pi(A, b, \unlhd ) \backslash \{X\} } 
Q\left(\frac{||X  s||^2 - (\tilde X  s)^T X  s }{\sigma ||\tilde X  s - X  s ||} \right) \\
&=& \sum_{\tilde X \in \Pi(A, b, \unlhd ) \backslash \{ X \}} 
Q\left(\frac{ ||\tilde X s - X  s  || }{2\sigma} \right)
\end{eqnarray}
based on a similar argument. The above equality holds since 
$
||X  s || = || \tilde X  s||
$
holds for any $\tilde X \in \Pi(A, b, \unlhd )$. Note that 
this simplification cannot apply to $\tilde X$ if  $\tilde X$ is a fractional vertex.
This is because the preservation of Euclidean norm does not hold in general for 
a doubly stochastic matrix. For example, we have
\begin{equation}
\left| \left|\left(
\begin{array}{ccc}
1/3 &  2/3 &  0 \\
0 & 1/3 & 2/3 \\
2/3 & 0 & 1/3 \\
\end{array}
\right) s \right | \right | = 1.9147 \ne ||s|| = \sqrt{5},
\end{equation}
where $s = (0,1,2)^T$.

If $\Pi(A,  b, \unlhd )$ have a group structure under the matrix multiplication, the above upper bound can be further simplified as
\begin{equation} \label{2ndupper}
P_{ML} \le 
\sum_{\tilde X \in \Pi(A, b, \unlhd ) \backslash \{I\}} 
Q\left(\frac{ ||\tilde X s -  s  || }{2\sigma} \right).
\end{equation}
It should be remarked that the second upper bound (\ref{2ndupper}) is independent of the transmitted codeword.
In order to prove the bound (\ref{2ndupper}), it is sufficient to prove $\Pi(A,  b, \unlhd )$ is 
distance invariant with respect to the Euclidean distance.

In the following, the distance invariant property of $\Pi(A,  b, \unlhd )$ will be shown.
Let us define the Euclidean distance enumerator by
\begin{equation}
W_{X}(Z) \defeq \sum_{\tilde X \in \Pi(A, b, \unlhd )} Z^{||X  s - \tilde X  s||}.
\end{equation}
This enumerator has the information on distance distributions measured from the permutation matrix $X$.

The next lemma states that the Euclidean distance enumerator does not depend on the center point $X$
if the linearly constrained permutation matrices have a group structure. This property can be regarded as 
a {\em distance invariance property} of permutation codes.
\begin{lemma}[Distance invariance]
If $\Pi(A,  b, \unlhd )$  forms a group under the matrix multiplication over $\Bbb R$,
the equality 
\begin{equation}
W_{X}(Z) = W(Z)
\end{equation}
holds for any  $X \in \Pi(A,  b, \unlhd )$.
The weight enumerator $W(Z)$ is defined by
$
W(Z) = W_{I}(Z)
$
where $I$ is the $n \times n$ identity matrix.
\end{lemma}
\begin{proof}
Since $\Pi(A,  b, \unlhd)$ forms a group,
the inverse $X^{-1}$ belongs to $\Pi(A,  b, \unlhd)$ as well. 
Since the inverse $X^{-1}$ induces a symbol-wise permutation,
it is evident that 
\begin{equation} \label{distinv}
||X  s - \tilde X  s|| = ||X^{-1} X  s - X^{-1}\tilde X  s|| =|| s - X^{-1}\tilde X  s||
\end{equation}
holds for any $X, \tilde X \in \Pi(A,  b, \unlhd) (X \ne \tilde X)$.
The Euclidean distance enumerator can be rewritten as
\begin{eqnarray} \nonumber
W_{X}(Z) &=& \sum_{\tilde X \in \Pi(A,  b, \unlhd)} Z^{||X  s - \tilde X  s||} \\ \nonumber
&=& \sum_{\tilde X \in \Pi(A, b, \unlhd)} Z^{|| s - X^{-1}\tilde X  s||} \\ 
&=& \sum_{X' \in \Pi(A,  b, \unlhd)} Z^{|| s - X'  s||}  = W(Z).
\end{eqnarray}
The second equality is a consequence of Eq. (\ref{distinv}).
The last equality is due to the assumption that $\Pi(A,  b, \unlhd)$ forms a group.
\end{proof}

\begin{example} \label{snr}
We have performed the following computer experiment for the following two codes:
\begin{enumerate}
\item LP decodable permutation code corresponding to the derangements of length 5. The additional 
linear constraint is $\trace(X)=0$. A transmitted word $(1,0, 4, 2, 3)^T$ is assumed.
The code polytope has 44-vertices which are all integral vertices.
\item LP decodable permutation code of length 5 corresponding to an additional linear constraint $X_{1,1} + X_{5,5} = 1$.
A transmitted word $(0,4,3,2,1)^T$ is assumed. The code polytope has 330-vertices. The set of vertices contains 
36-integral vertices and 294-fractional vertices.
\end{enumerate}
The AWGN channel with 
noise variance $\sigma^2$ is assumed.  The signal-to-noise ratio is defined by
$
SNR = 10 \log_{10}\left(1/{\sigma^2} \right).
$
The LP decoding described in the previous section was employed for decoding.

Figure \ref{BER-LP} presents the upper bounds and simulation results
on block error probability of these permutation code.
It is readily observed that the upper bounds presented in this section shows 
reasonable agreement with the simulation results.

The both codes have the same minimum pseudo distance $0.707107$ and 
similar cardinalities (44 and 36) but the derangement code provides much better block error probabilities than those 
of the code with the constraint $X_{1,1} + X_{5,5} = 1$. This is because the existence of fractional 
vertices (i.e., 294-fractional vertices) severely degrades the decoding performance of the code with the constraint $X_{1,1} + X_{5,5} = 1$
compared with the derangement code.
\hfill\qed
\end{example}

\begin{figure}[htbp]
\begin{center}
\includegraphics[scale=0.9]{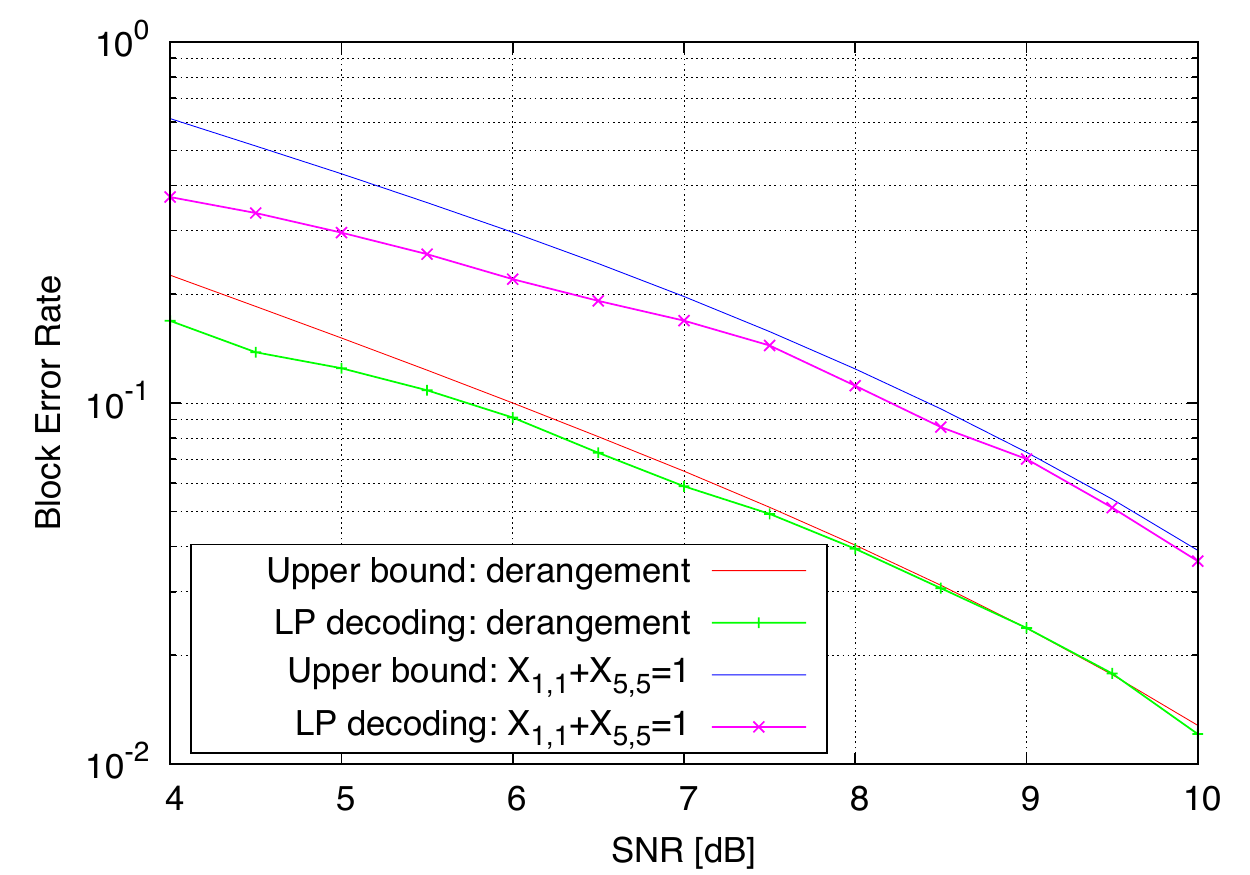}
\caption{Comparison of upper bounds and simulation results for LP decoding on block error probabilities $(n=5)$}
  \label{BER-LP}
\end{center}
\end{figure}

\section{Some classes of linearly constrained permutation codes}
\label{structured}
In this section, we will discuss some classes linearly constrained permutation codes which 
are easy to encode.

\subsection{Repetition permutation codes}

Let $\eta$ be a positive integer.
Assume that a positive integer $n$ is a multiple of $\eta$.
The {\em repetition permutation codes} with repetition order $\eta$ is defined by
\begin{equation}
\{ ((Ys_1)^T, (Y s_2)^T, \ldots, (Y s_\eta)^T )^T \in \Bbb R^n \mid  Y \in \Pi_{n/\eta}  \},
\end{equation}
where $s_1, s_2, \ldots, s_\eta \in \Bbb R^{n/\eta}$. 
We here assume that  all the elements in $s_1, \ldots, s_\eta$ are distinct each other.
It is evident that the 
cardinality of the code is given by $(n/\eta)!$.
The minimum Hamming distance of the code is $2 \eta$ because the minimum Hamming distance of 
$Y s_i$ is 2 for any $i \in [1,\eta]$.

It should be remarked that the repetition permutation code is a linearly constrained permutation 
code. The next example demonstrate linear constraints for the repetition permutation codes.
\begin{example}
Let 
\[
X=
\left(
\begin{array}{cccc}
X_{1,1} & X_{1,2} & X_{1,3} & X_{1,4} \\
X_{2,1} & X_{2,2} & X_{2,3} & X_{2,4} \\
X_{3,1} & X_{3,2} & X_{3,3} & X_{3,4} \\
X_{4,1} & X_{4,2} & X_{4,3} & X_{4,4} \\
\end{array}
\right).
\]
The permutation matrices in $\Pi_4$ satisfying the following 
set of linear constraints
\begin{eqnarray}
X_{1,3} = X_{1,4} = X_{2,3} = X_{2,4} = 0 \\
X_{3,1} = X_{3,2} = X_{4,1} = X_{4,2} = 0 \\
X_{1,1}=X_{3,3},\ X_{1,2} = X_{3,4} \\
X_{2,1}=X_{4,3},\ X_{2,2} = X_{4,4}
\end{eqnarray}
defines the repetition permutation code of length 4 with repetition order 2.
\hfill \fbox{}
\end{example}

\subsection{Cartesian product codes}

Suppose that $\eta$ is a positive number and that $n$ is positive multiple of $\eta$.
A set of permutation matrices $U \subset \Pi_{n/\eta}$ is assumed to be given.
The {\em cartesian product codes} is defined by
\begin{equation}
\{ ((Y_1 s_1)^T, (Y_2 s_2)^T, \ldots, (Y_\eta s_\eta)^T )^T \in \Bbb R^n \mid  Y_1,\ldots, Y_\eta \in U  \},
\end{equation}
where $s_1, s_2, \ldots, s_\eta \in \Bbb R^{n/\eta}$. The cardinality of cartesian product codes
is thus given by $|U|^\eta$ if all the elements in $s_1, \ldots, s_\eta$ are distinct each other.
Note that the class of cartesian product codes can be defined based on a set of linear constraints as well
if $U$ is defined by linear constraints.

\subsection{Pure involution codes}

In this subsection, we focus on the set of pure involutions, which produces a non-trivial class of 
permutation codes.
It will be shown that the class of the permutation codes defined based on the pure involutions 
possess several good properties.  This class of code can be encoded with an efficient greedy encoding algorithm.
The cardinality of the code is much larger than the repetition code with the same length and the same minimum 
Hamming distance. 

An {\em involution} is a permutation which coincides with its inverse permutation.
Namely, the necessary and sufficient condition for a permutation matrix $X\in \Pi_n$ to be an involution
is $X = X^T$ because the inverse matrix of a permutation matrix is the transposition of it.
A {\em pure involution}  is an involution without fixed point;
i.e.,  a permutation matrix $X\in \Pi_n$  is said to be a pure involution 
if and only if $X = X^T$ and $\trace(X) = 0$. In other words, the set of pure involutions is the intersection 
of the set of involutions and the set of derangements. 

A pure involution exists when $n$ is a positive even number. The reason is as follows.
The lower triangle below the diagonal of $X$ and 
the upper triangle above the diagonal must have the same number of ones since $X=X^T$.
This implies that the number of ones in $X$ should be  even since the diagonal is constrained to be zero.
A permutation matrix $X \in \Pi_n$ contains $n$-ones. Thus, if $n$ is odd, it is clear that no permutation matrix meets 
the constraints. Throughout this subsection, we assume that $n$ is an even positive number.

Let 
\[
\Omega_n \defeq \{X  \in \Pi_n \mid X=X^T, \trace(X)=0 \}.
\]
It is known that the cardinality of the pure involutions is given by
\begin{equation}
|\Omega_n|= (n-1) (n-3) \times \cdots \times 3 \times 1=  \frac{n!}{2^{n/2}  (n/2)!}.
\end{equation}

The linearly constrained permutation codes defined based on the constraints 
$X=X^T, \trace(X)=0$ is called the {\em pure involution codes}.
The triple for the pure involution codes are given by
\begin{equation}
A=\left(
\begin{array}{c}
\sfvec(I_n) \\
\sfvec\left(F^{(2,1)} \right) \\
\sfvec\left(F^{(3,1)} \right) \\
\vdots \\
\sfvec\left(F^{(n,n/2-1)} \right) \\
\end{array}
\right),\quad 
b= \ve 0, \quad
\unlhd = (=,\ldots, =)^T,
\end{equation}
where $F^{(i,j)} \in \{0,1\}^{n \times n}$ is the binary matrix defined by
\[
F_{a,b}^{(i,j)} = 
\left\{
\begin{array}{cc}
1, & (a,b)=(i,j) \\
-1, & (a,b)=(j,i) \\
0, & \mbox{otherwise}.
\end{array}
\right.
\]

\subsubsection{Greedy encoding algorithm for pure involutions}

A significant advantage of the pure involutions is that there exists 
an efficient encoding algorithm.  The procedure {\sf EncMap} shown below
can be considered as 
a greedy algorithm for a constraint satisfaction problem without a back-tracking process.

\noindent{\underline{\sf EncMap}}
\begin{enumerate}
\item[*] Input: $m \in [1, (n-1)\times (n-3) \cdots 3\times 1]$ (message)
\item[*] Output:  $X \in \Omega_n$ (pure involution)
\item $m := m-1;$
\item for ($p := 0$; $p < n/2$; $p:=p+1$) \{
\item \quad $a_p := [m \mbox{ mod } (2p+1)]+1;$
\item \quad $m := m  \ {\sf  div}\   (2p+1)$;
\item \}
\item $\forall i,j \in [1,n], X_{i,j} := 0$;
\item $\forall i,j  \in [1,n] (i \ne j), Z_{i,j} := 1$; $\forall i \in [1,n], Z_{i,i} := 0$;
\item for ($p := n/2-1$; $p \ge 0$; $p:=p-1$) \{
\item \quad $j := \arg \min \left\{j' \in [1,n]: \sum_{i'=1}^n Z_{i',j'} > 0  \right\} $;
\item \quad $i := \arg \min \left\{k \in [1,n]: \sum_{i'=1}^k Z_{i',j} = a_{p}  \right\} $;
\item \quad $X_{i,j} := 1$; $X_{j,i} := 1$;
\item \quad $\forall q \in [1,n]$, $Z_{q ,j} := 0$, $Z_{j, q} :=0$, $Z_{i,q} := 0$, $Z_{q, i} :=0$;
\item \}
\item Output $X$;
\end{enumerate}

The arithmetic operation in the line 4 represents the division for integers; i.e., 
$5 \ {\sf div}\ 2 = 2$.
There are some remarks on {\sf EncMap}.
The part from the line 1 to 5 converts a message integer into an $n/2$-tuples:
\[
(a_0, a_1, \ldots, a_{n/2-1}) \in [1,1] \times [1,3] \times \cdots  \times [1, n-1].
\]
The remaining part generates a pure involution according to the $n$-tuple $(a_0, a_1, \ldots, a_{n/2-1})$.

The variables $Z_{i,j}$ represents whether $X_{i,j}$ is determined $(Z_{i,j} = 0)$  or not $(Z_{i,j} = 1)$. 
On the diagonal elements of $Z_{i,j}$ are initialized to be zero which means that the diagonal elements of 
$X_{i,j}$ is determined to be zero.

The generation of a pure involution is performed  in a greedy manner from the left columns to the right columns.
The undetermined column with the smallest index is found in the line 9.
In the line 10,  the row index of $a_p$-th undetermined element is assigned  to $i$.
In the line 11,  two ones are written at $(i,j)$ and $(j,i)$-positions of $X$ and the line 12 fixes 
the cross regions around $(i,j)$ and $(j,i)$.

In an encoding process, for any $p=N-t (t \in [1,N])$,
\begin{equation}
\sum_{i' \in [1,n]}^n Z_{i',j} = (n-1) - 2 (t-1) = 2p+1
\end{equation}
holds at the line 10.
This is because exactly two-columns and two-rows of $Z$ are set to zero for each iteration
due to the constraints of the pure involution.
In other words, the numbers of zero columns and zero rows are increased by two after an 
iteration. This property guarantees that
\begin{equation}\label{injection}
\sum_{i' \in [1,n]}^n Z_{i',j} \ge a_{p}
\end{equation}
holds for all $p \in [0,N-1]$. Therefore, for any input $m$, the line 10 can find 
an index $i$ satisfying 
\[
i= \arg \min \left\{k \in [1,n]: \sum_{i'=1}^k Z_{i',j} = a_{p}  \right\}.
\]

The loop from the line 2 to 5 takes $O(n)$-time
under the assumption that the basic big-number arithmetics can be done within a unit time.
The initialization process (lines 6 and 7) requires $O(n^2)$-time.
The most time consuming part of {\sf EncMap} is the loop from the line 8 to 13.
In order to find $i,j$ in lines 9 and 10, $O(n)$-times requires. The process in line 12 also needs
$O(n)$-time to carry it out.
Therefore, the time complexity of the loop (from the line 8 to 13.) is $O(n^2)$, which dominates the
time complexity of {\sf EncMap}.

From the definition shown above, it is evident that  {\sf EncMap} gives 
a injection map from $[1, (n-1)\times (n-3) \cdots 3\times 1]$ to $\Omega_n$.
Since the cardinality  of  $\Omega_n$ is $(n-1)\times (n-3) \cdots 3\times 1$, 
we can see that  {\sf EncMap} is a bijection.

There is an inverse map of {\sf EncMap} from $\Omega_n$ to $[1, (n-1)\times (n-3) \cdots 3\times 1]$
because {\sf EncMap} is a bijection.   The procedure {\sf DecMap} gives the inverse map of {\sf EncMap}.

\noindent{\underline{\sf DecMap}}
\begin{enumerate}
\item[*] Input: $X \in \Omega_n$ (pure involution)
\item[*] Output: $m \in [1, (n-1)\times (n-3) \cdots 3\times 1]$ (message)
\item $\forall i,j  \in [1,n] (i \ne j), Z_{i,j} := 1$; $\forall i \in [1,n], Z_{i,i} := 0$;
\item for ($p := n/2-1$; $p > 0$; $p:=p-1$) \{
\item \quad $j := \arg \min \left\{j' \in [1,n]: \sum_{i'=1}^n Z_{i',j'} > 0  \right\} $;
\item \quad $i := \sum_{i' \in [1,n]} i' \Bbb I[X_{i',j} = 1]$;
\item \quad $a_p:= \sum_{i'=1}^i Z_{i', j}$;
\item \quad $\forall q \in [1,n]$, $Z_{q ,j} := 0$, $Z_{j, q} :=0$, $Z_{i,q} := 0$,  $Z_{q, i} :=0$;
\item \}
\item $m := 0;$
\item for ($p := n/2-1$; $p \ge 1$; $p:=p-1$) \{
\item \quad $m := (2p+1) m+ (a_p-1)$;
\item \}
\item $m:=m+1$;
\item Output $m$;
\end{enumerate}

\begin{example}
An encoding process of a pure involution matrix is illustrated in Fig.\ref{encodingprocess}.
In this example, $n=6$ is assumed. The status of $X_{i,j}$ and $Z_{i,j}$ are depicted by $6 \times 6$ cells in 
Fig.\ref{encodingprocess}. Namely, $Z_{i,j} = 1$ (undetermined state) represents an empty cell. 
A cell with label 0 (resp. 1) represents $(X_{i,j}, Z_{i,j}) = (0,0)$ (resp. $(X_{i,j}, Z_{i,j}) = (1,0)$).
At first, the diagonal cells are set to be zero because of the constraint $\trace(X)=0$.
The message is assumed to be $m=5$. In this case, we have $a_0=1, a_1=3, a_2=2$.
The shaded cells in Fig.\ref{encodingprocess} (a) represents possible places to write the symbol 1.
According to the part of the message $a_2=2$, the second shaded cell is determined to be 1.
In Fig.\ref{encodingprocess} (b), the symbol 1 is written on the symmetric position and zeros are placed in 
the columns and rows corresponding to two 1's. In a similar way (Fig.\ref{encodingprocess} (b)--(e)),
the empty cells are filled with 0 or 1. As a result, we have a pure involution matrix (Fig.\ref{encodingprocess} (f)).
\end{example}

\begin{figure}[htbp]
\begin{center}
\includegraphics[scale=0.4]{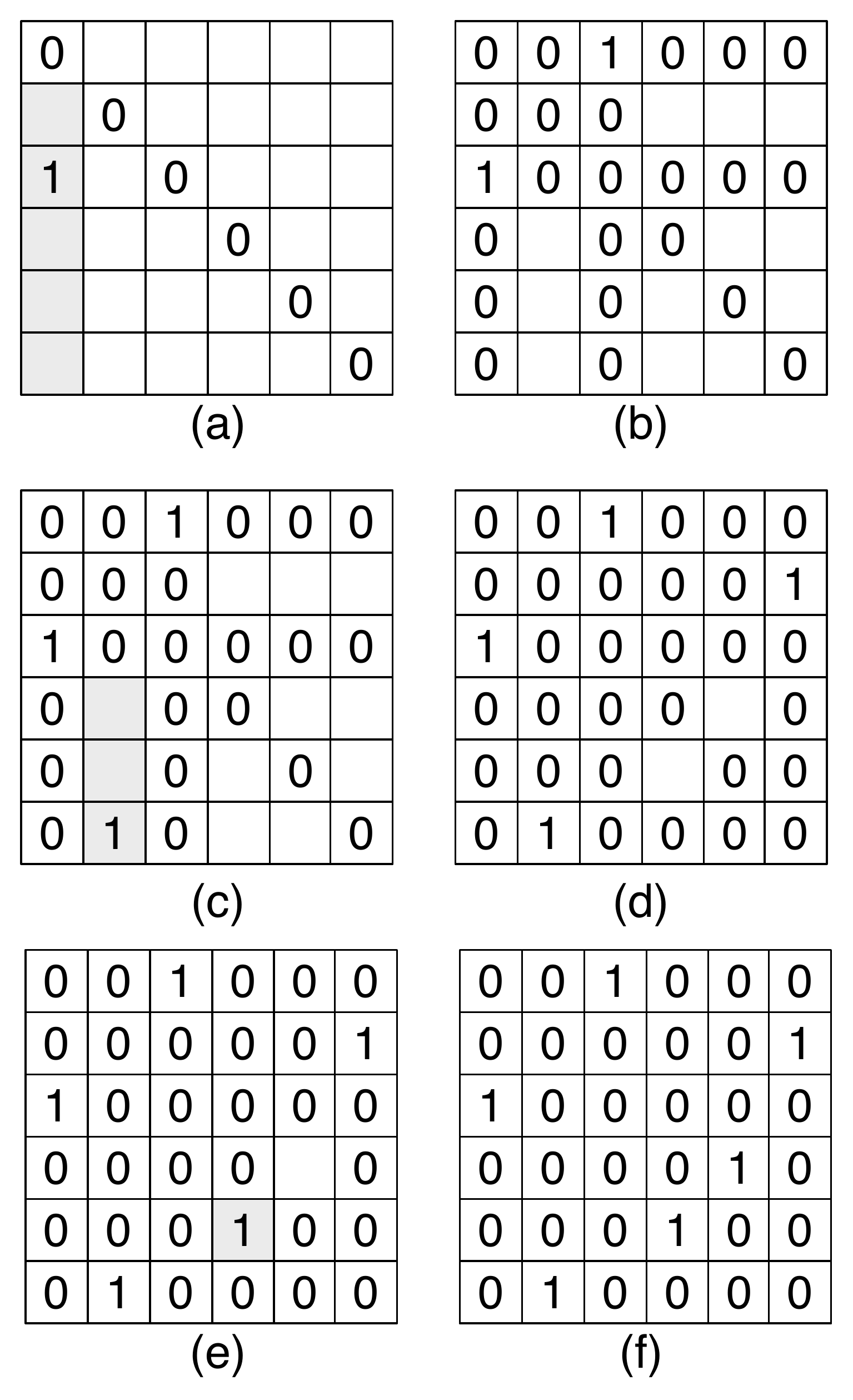}
\flushleft{{\scriptsize The shaded cells represent are possible places to write the symbol 1.
In (a) and (b), there are 5 and 3-shaded cells, respectively. This means that 
$5 \times 3 = 15$ pure involution matrices exist when $n=6$.}}
\caption{An Encoding process of a pure involution matrix}
  \label{encodingprocess}
\end{center}
\end{figure}

\subsubsection{Minimum Hamming distance of pure involution codes}

Let $s \in \Bbb R^n$ be an initial vector whose components are distinct each other.
It is well known that the minimum Hamming distance of $\Lambda(s)$ is given by
\begin{equation}
\min_{X,X' \in \Pi_n (X  \ne X')}d_H(Xs ,X' s) = 2.
\end{equation}

The minimum Hamming distance of the pure involution codes is larger than that of $\Lambda(s)$.
\begin{lemma}[Minimum distance]
The minimum Hamming distance of the pure involution codes are given by
\begin{equation}
\min_{X,X' \in \Omega_n (X  \ne X')}d_H(Xs ,X' s) = 4.
\end{equation}
\end{lemma}
\begin{proof}
Assume that $X,X' \in \Omega_n (X  \ne X')$.
Since $X \ne X'$, there is an index pair $(i,j) \in [1,n]^2$ 
satisfying $X_{i,j} \ne X'_{i,j}$.
Without loss of generality, we assume that $X_{i,j} =1$ and $X'_{i,j}=0$.
An index $l \in [1,n]$ satisfying $X_{i,l} \ne X'_{i,l}$ must exist 
because $X$ and $X'$ are permutation matrices. Due to the assumption
$X_{i,j} =1$ and $X'_{i,j}=0$, we have $X_{i,l} = 0$ and $X'_{i,l} = 1$.
In a similar manner, there must be an index $k$ satisfying $X_{k,j}=0, X'_{k,j}=1$.
It is possible to continue this argument until a sequence of index pairs constitutes a loop.

The set of the index pairs
$
\{(i,j) \in [1,n]^2 \mid X_{i,j} \ne X'_{i,j} \}
$
is called a {\em difference position set}. The argument above implies that 
the difference position set needs to be partitioned into several loops of even length. A loop 
means a sequence of adjacent index pairs with the form 
$(i_1,i_2)  \rightarrow (i_1,i_3) \rightarrow (i_4,i_3)  \rightarrow \cdots  \rightarrow (i_1,i_2)$.
If  $X_{i_1,i_2} = 1$ holds, then we have $X_{i_1,i_3} = 0, X_{i_4,i_3} = 1$ and so on.
Therefore, the length of a loop should be even because a loop with odd length gives inconsistent assignment 
$X_{i_1,i_2} = 0$ at the end of the loop.

The shortest loop of even length have the form  $(i,j) \rightarrow (i,l) \rightarrow (k,l) \rightarrow (k,j) \rightarrow (i,j)$.
If the difference potion set includes this type of a loop of length 4, it must also contain another loop of length 4 
with the form $(j,i) \rightarrow (l,i) \rightarrow (l,k) \rightarrow (j,k) \rightarrow (j,i)$ because 
$X = X^T$ holds for any $X, X' \in \Omega_n$ (See Fig.\ref{twoloops}). 
Let $a=X s$ and $a'=X' s$. If the difference position sets consist of only such two symmetric loops of length 4,
we have
\[
a_u \ne a'_u\quad \mbox{iff } u \in \{i, j, k, l \}.
\]
This implies that the smallest number of differences between 
$X s$ and $X' s$ is 4. 
\end{proof}
The proof of the above lemma indicates a way to enumerate the number of codewords at the minimum Hamming distance.
For a fixed $X s$, the number of codewords $X' s$ satisfying 
$
d_{H}(Xs , X' s) = 4
$
can be obtained by enumerating the number of allocations of two symmetric loops.

\begin{figure}[htbp]
\begin{center}
\includegraphics[scale=0.5]{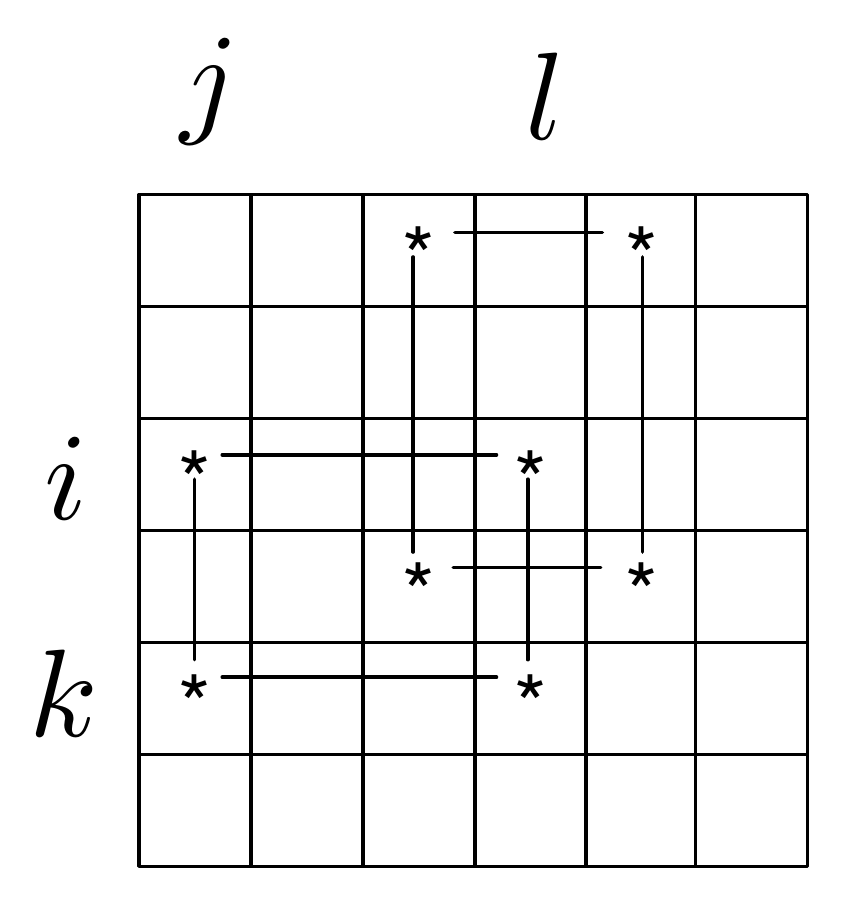}
\flushleft{{\scriptsize The left loop of length 4 represents
$(i,j) \rightarrow (i,l) \rightarrow (k,l) \rightarrow (k,j) \rightarrow (i,j)$ and 
the right loop corresponds to 
$(j,i) \rightarrow (l,i) \rightarrow (l,k) \rightarrow (j,k) \rightarrow (j,i)$.
Note that there are 4-columns which include elements of the difference position set.
These columns correspond to the positions on which the symbols of $X s$ and $X' s$ differ.
}}
\caption{Two symmetric loops of length 4 in a difference position set.}
  \label{twoloops}
\end{center}
\end{figure}

We have seen that the repetition code of repetition order 2 
yields the minimum Hamming distance 4. When the length of the code is $n$ (even), 
the number of codewords is given by $(n/2)!$. On the other hand,  the pure involution code 
provides the same minimum Hamming distance and the cardinality of the code is given by
${n!}/({2^{n/2}  (n/2)!})$, which is much larger than $(n/2)!$ because 
\begin{eqnarray}
\frac{{n!}/({2^{n/2}  (n/2)!})}{(n/2)!} &=& {n \choose n/2} 2^{-n/2} \simeq \frac{1}{\sqrt{\pi n/2}}2^{n/2}. 
\end{eqnarray}

For example, consider the case where $n=64$.
In this case, the number of codewords of the repetition code is  $(n/2)! \simeq 2^{118}$.
On the other hand, the pure involution code have 
\[
\frac{n!}{2^{n/2}  (n/2)!} \simeq 2^{146}
\]
codewords which is approximately $2^{28}$-times larger than that of the repetition code.

\subsubsection{Code polytope of pure involutions}
The linear constraint $X=X^T$ and $\trace(X)=0$ for pure involutions defines a 
code polytope which is not an integral polytope. 
\begin{example}
Assume that $n=6$. The code polytope defined based on the constraints 
$X=X^T$ and $\trace(X)=0$ have 15 integral vertices and 10 fractional vertices.
A fractional vertex is 
\[
\left(
\begin{array}{cccccc}
0 &1/2 &0 & 0 & 0 & 1/2 \\
1/2 & 0 & 0 & 0  &0 & 1/2  \\
0 &0 & 0& 1/2 &1/2 & 0 \\
0 & 0& 1/2& 0 &1/2&  0  \\
0 &0 &1/2 &1/2 &0 &0  \\
1/2 &1/2& 0 &0& 0& 0  \\
\end{array}
\right).
\]
\end{example}
Deriving inequality description of the convex hull of pure involution matrices 
is an interesting open problem.

\subsubsection{Simulation results}
The minimum Hamming distance of a permutation code is a universal measure 
for goodness of a code because it does not depend on the choice of the initial vector $s$.
However, as we have seen in the previous section,  decoding performance 
is mostly determined by the pseudo distance distribution of a code polytope.

In order to evaluate the decoding performance of pure involution codes, 
we have performed a computer experiment.
Figure \ref{pureinvsim} presents the block error probability of the pure involution codes with 
length 64. In this experiment, the initial vector is assumed to be $s=(1,2,\ldots, 64)$ and
the LP  decoding was used. The definition of 
the SNR is the same as in Example \ref{snr}. For comparison purpose, the block error probabilities of 
the repetition permutation code of length 64 with the repetition order 2 and uncoded 
permutations vectors (i.e., $\Lambda(s)$) of length 64  are also plotted in Fig. \ref{pureinvsim}.
It can be observed that the pure involution code gives much small block probabilities compared with 
the repetition code.  As we have seen  in the previous section, the cardinality of a pure involution code 
is much larger than that of the repetition code.  We may be able to conclude that the pure involution code
is superior to the repetition code.
\begin{figure}[htbp]
\begin{center}
\includegraphics[scale=0.9]{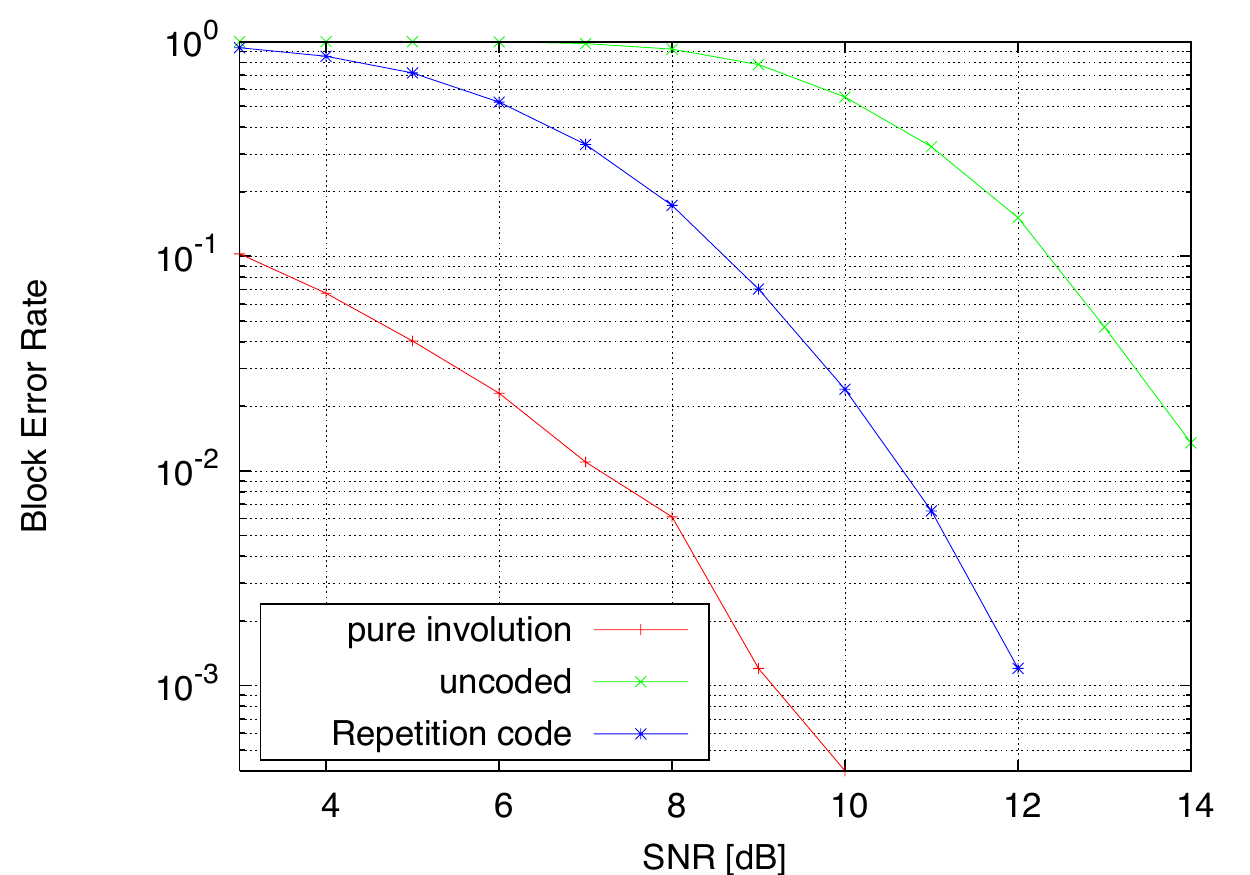}
\caption{Comparison of block error probabilities: pure involution codes, repetition permutation codes, and
uncoded permutation vectors of length 64}
  \label{pureinvsim}
\end{center}
\end{figure}

\subsection{Block permutation codes}
A block permutation codes are defined based on the block permutation matrices.
The block structure is useful for encoding and evaluation of the minimum squared Euclidean distance.

\subsubsection{Definitions}
Suppose the situation where the set $[1, n] \times [1,n]$ is divided into 
mutually disjoint $\gamma \times \gamma$ square blocks of size $\nu \times \nu$ (i.e., $n=\gamma \nu$ holds).
The square blocks are called {\em blocks} which is explicitly defined as follows.
\begin{definition}[Block]
For $k, b \in [1 , \gamma]$, a {\em block} $B_{k,b}$  is defined by
\begin{equation}
 B_{k, b} \defeq \{ (i, j) \in [1,n]^2 \mid  \nu(k-1) < i \le \nu k, \nu(b-1) < j \le \nu b \}.
 \end{equation}
The indices $k$ and $b$ are called {\em block indices}.
\hfill\fbox{}
\end{definition}

The rectangle region $T^{(l)}_{k, b}$ is defined as 
\begin{equation}
T^{(l)}_{k,b} \defeq \{ (x, y) \in B_{k,b} \mid  y = \nu(b-1) + l \}
\end{equation}
for $k, b \in [1, \gamma]$ and $l \in [1, \nu]$.
The subscript $k,b$ specifies the block where the rectangle region $T^{(l)}_{k, b}$ belongs to.
The superscript  $l \in [1, \nu]$, which is called a {\em subindex}, indicates the relative position 
in the block  $B_{k, b}$. 

We are now ready to define a block permutation matrix which is the basis for realizing a block-wise permutation group.
\begin{definition}[Block permutation matrix]
Assume that a permutation matrix $X \in  \Pi_n$ is given.
If, for any $b \in [1,\gamma]$, there exists the unique block index $k$
satisfying 
\begin{equation}
X(B_{k, b}) \neq 0
\end{equation}
then $X$ is called a {\em block permutation matrix}.
The notation $X(B_{k, b})$ represents the sub-matrix of $X$ corresponding to the block $B_{k, b}$. \hfill\fbox{}
\end{definition}

From this definition, it is apparent that 
a nonzero $X(B_{k, b}) \in \{0,1\}^{\nu \times \nu}$ is a permutation matrix if $X$ is a block permutation matrix.
Furthermore, there exists the unique block index $b$ satisfying $X(B_{k, b}) \neq 0$
for any block index $k \in [1,\gamma]$. This equivalent statement can be obtained 
by exchanging the role of column and row in the above definition.

For block indices $k, b \in [1, \gamma]$ and subindex $ l \in [1, \nu]$,  the {\em skewed column set} is defined by 
\begin{equation}
 U_{k,b}^{(l)} \defeq 
 T_{k,b}^{(l)} \cup \left( \bigcup_{k' \in [1,\gamma] \backslash \{k\}} T_{k' , b}^{(l \mod \nu) + 1  }   \right).
\end{equation}

Figure \ref{blocks} illustrates the subsets of $[1,n] \times [1,n]$ appeared so far 
such as the blocks, the rectangle regions, and the skewed column set.
\begin{figure}[htbp]
\begin{center}
\includegraphics[scale=0.4]{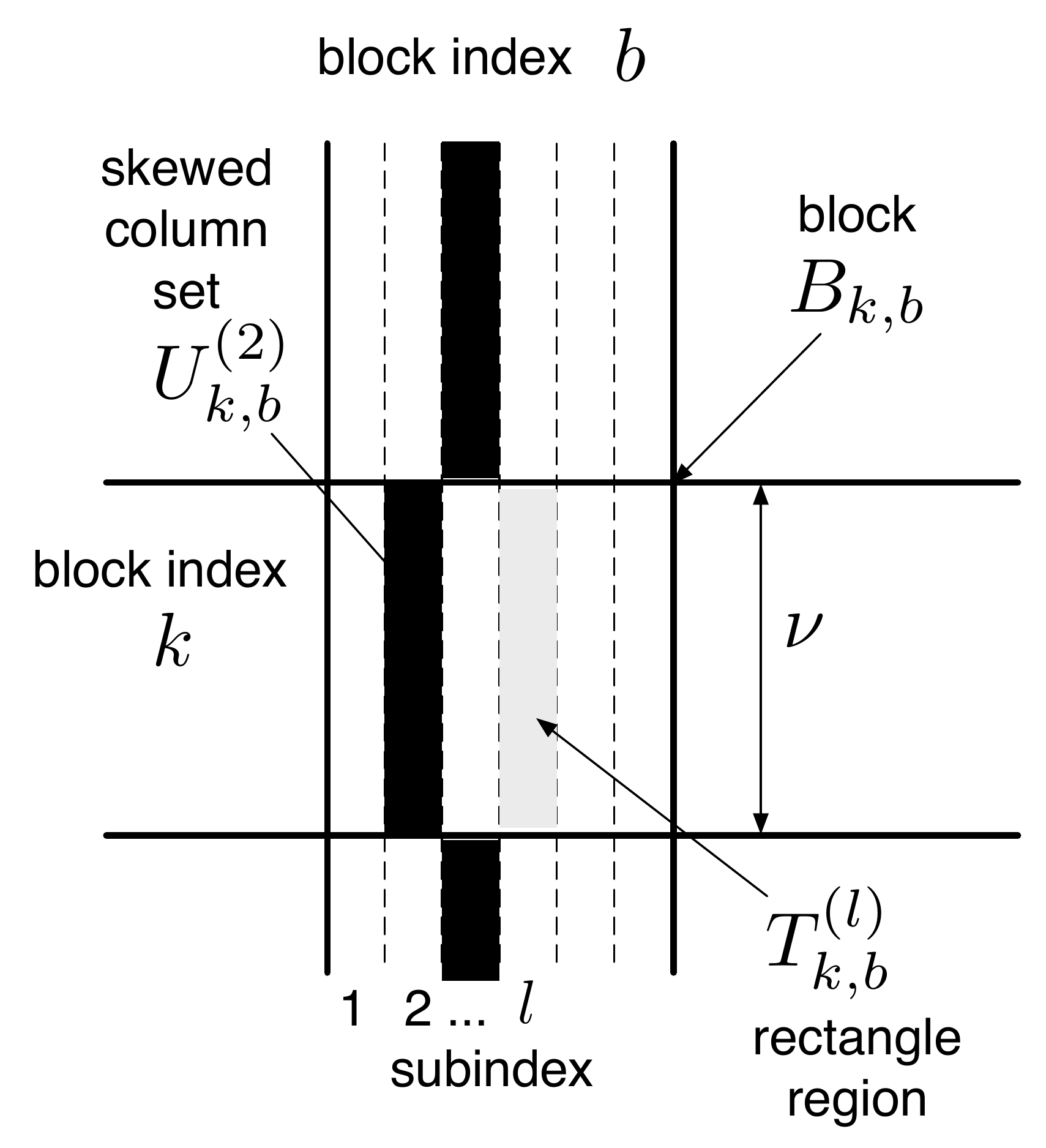}
\caption{Blocks, rectangle regions and skewed column set}
  \label{blocks}
\end{center}
\end{figure}

\subsubsection{Block permutation codes}
The next theorem presents a set of linear constraints characterizing block permutation matrices.
\begin{theorem}[Characterization of block permutation matrix]
\label{characterization}
Let $X \in \Pi_n$ be a permutation matrix.
The permutation matirx $X$ is a block permutation matrix if and only if 
\begin{equation} \label{hagiwara}
\sum_{(u,v) \in U_{k,b}^{(l)}} X_{u,v} = 1
\end{equation}
holds for any $b,k\in [1, \gamma], l \in [1, \nu]$.
\end{theorem}

The next example clarifies the linear constraints characterizing a $4 \times 4$ block permutation matrix.
\begin{example}
Let $n = 4, \nu =2, \gamma=2$.
The necessary and sufficient condition for a permutation matrix $X \in \Pi_4$ being 
a block permutation matrix are as follows:
\begin{eqnarray} \nonumber
X_{1,1} + X_{2,1} + X_{3,2} + X_{4,2}  &=& 1 \\ \nonumber
X_{1,2} + X_{2,2} + X_{3,1} + X_{4,1}  &=& 1 \\ \nonumber
X_{1,3} + X_{2,3} + X_{3,4} + X_{4,4}  &=& 1 \\ \nonumber
X_{1,4} + X_{2,4} + X_{3,3} + X_{4,3}  &=& 1. \\ \label{subblock1}
\end{eqnarray}
\end{example}

Let us denote the set of block permutation matrices by
\begin{equation}
\Pi(n, \nu) \defeq \{X \in \Pi_n \mid X \mbox { satisfies } (\ref{hagiwara})    \}.
\end{equation}
Note that we here employ a lighter notation $\Pi(n, \nu)$ instead of $\Pi(A,  b, \unlhd)$ since it 
explicitly express dependency on $n$ and $\nu$.
It should be remarked that $\Pi(n, \nu)$ forms a group under matrix multiplication over $\Bbb R$.

The class of block permutation codes defined below is a class of LP decodable permutation codes.
\begin{definition}[Block permutation code]
Let $n$ be a positive integer. A positive integer $\nu$ is a divisor of $n$.
The initial vector $ s$ belongs to $\Bbb R^n$.
The {\em block permutation code} $C(n, \nu,  s)$ is defined by 
\begin{equation} \label{ineq}
C(n, \nu,  s) \defeq \{ X  s \in {\Bbb R }^n : X \in \Pi(n, \nu)  \}.
\end{equation}
\hfill\fbox{}
\end{definition}

In Section \ref{decodinganalysis}, we saw the minimum pseudo distance is 
one of most influential parameters for LP decoding performance.
Unfortunately, the evaluation of the minimum pseudo distance is not a trivial problem.
As a possible alternative, we here evaluate 
the minimum squared Euclidean distance of $C(n, \nu,  s)$ defined by 
\begin{equation}
d_{min}^2 \defeq \min_{ x,  y \in C(n, \nu,  s) ( x \ne  y)} || x -  y||^2.
\end{equation}
At least, we can say that  decoding performance degrades  even with an ML decoder 
if $C(n, \nu,  s)$ has small $d_{min}^2$.

The block-wise permutation structure of a block permutation code 
can be exploited for deriving a simple formula 
on the minimum squared Euclidean distance.

Let  us define $\Delta_1^2$ and $\Delta_2^2$ by
\begin{eqnarray} \nonumber
\Delta^2_1 &=&\min_{k \in [1, \gamma] }  
\min_{Q  \in \Pi_\nu (Q \ne I)}  || s_{k} - Q  s_{k}  ||^2 \\
\Delta^2_2 &=& \min_{k,j \in [1, \gamma] (k \ne j) } \min_{Q \in \Pi_\nu  } ||s_{k} - Q s_{j}  ||^2.
\end{eqnarray}
Assume that  both $\Delta_1^2$ and $\Delta_2^2$ are positive for given $n, \nu, s$.
In such a case, $C(n, \nu,  s)$ is non-singular and it is easily proved that 
the minimum squared Euclidean distance of $C(n, \nu,  s)$ is given by
\begin{equation}\label{mind}
d^2_{min} =\min\{  \Delta_1^2,  2\Delta_2^2  \}.
\end{equation}

The following example illustrates that a block permutation code can have more codewords
than those of a trivial cartesian product code under the condition that both of codes have the same minimum 
squared Euclidean distance.
\begin{example}
Let $n=8, \gamma=2, \nu=4$. The initial vector $s=(s_1^T,s_2^T)^T$ is assumed to be 
\[
 s_1 = 
\left(
\begin{array}{c}
1 \\
3 \\
5 \\
7\\
\end{array}
\right), \quad
 s_2 = 
\left(
\begin{array}{c}
2 \\
4 \\
6 \\
8 \\
\end{array}
\right)
\]
From the definition of $\Delta_1^2, \Delta_2^2$, we easily obtain 
$
\Delta_1^2= 8, \quad \Delta_2^2  = 4.
$
From (\ref{mind}), we have 
$
d^2_{min} = \min \{8, 2\times 4 \} = 8.
$
The number of codewords is $\gamma ! \times (\nu!)^\gamma = 1152$.
The cartesian product code defined by
\[
\{ ((Y_1\  s_1)^T, (Y_2\  s_2)^T )^T \in \Bbb R^{64} \mid  Y_1, Y_2 \in \Pi_4  \},
\]
has also squared Euclidean distance 8 but it contains 576-codewords, 
which is half of the number of codewords of the block permutation code.
\hfill \fbox{}
\end{example}

\section{Randomly constrained permutation matrices}
\label{random}

In the previous section, we discussed a set of structured permutation matrices.
Another possible choice for linear constraints is to generate them randomly.
Such random linear constraints are amenable for probabilistic analysis and 
appears interesting from information theoretic view. In this section, we study a 
class of LP decodable permutation codes defined based on  random constraints.

\subsection{Sparse constraint matrix ensemble}

Since the LP decodable permutation codes are non-linear codes,
the cardinality of a given code cannot be determined directly from the constraints in general.
In the following part of this section, we will analyze the cardinality of codes and their
Hamming weight distributions.

A sparse constraint matrix ensemble is assumed 
in the following analysis, which has a close relationship to the analysis on average 
weight distribution of LDPC ensembles \cite{LS02}.

The linear constraint assumed here is the equality constraint for two variables 
such as $X_{i,j}=X_{k,l}$. As discussed in Section X, linearly constrained permutation 
matrices defined based on this equality constraint is important because such 
matrices can be used as building blocks of a generalized block permutation code.

Let $S$ be the set of binary constraint matrices:
\begin{equation}
S \defeq \{A \in  \{0,1\}^{m \times n^2}: \mbox{every row of $A$ contains $2$-ones}\}.
\end{equation}
We assign the uniform probability 
\begin{equation}
P(A) \defeq \frac{1}{{n^2  \choose 2}^m}
\end{equation}
to each matrix in $S$. The pair $(S, P)$ can be considered as 
an ensemble of matrices, which becomes the basis of the following probabilistic method.

Assume that $\theta: S \rightarrow \{-1, 0,1\}^{m \times n^2}$
is defined by $B = \theta(A)$, where 
\begin{equation}
B_{i,j} = 
\left\{
\begin{array}{cl}
-A_{i,j},  & \mbox{if } \forall j' \in [1,j-1], A_{i,j'} = 0,   \\
A_{i,j}, & \mbox{otherwise}.
\end{array}
\right.
\end{equation}
Note that $\theta(A) \sfvec (X) = 0$ corresponds to $m$ equality constraints of two variables.

In this section, we focus on the LP decodable permutation code 
$\Lambda(\theta(A),0,  \unlhd,  s)$, where $A \in S$ and 
$\unlhd \defeq (\overbrace{=,=  \ldots, =}^{m \mbox{\scriptsize}})^T$.
The symbol $\ve 1$ denotes the vector of length $m$ whose entries are all ones.
Extensions of the analysis for more general classes of LP decodable permutation codes are possible, but 
we here focus on the simplest class to explain the idea of the analysis.
Throughout this section, we assume that components of the initial vector $ s$ differ each other.

\subsection{Probabilistic analysis on average cardinality of codes}

The number of codewords in $\Lambda(\theta(A), 0, \unlhd,  s)$ is given by
\begin{equation}
M(A) \defeq \sum_{X \in \Pi_n} \Bbb{I} [\theta(A)\  \mbox{\sf vec}(X) \unlhd 0],
\end{equation}
where $\Bbb I$ is the indicator function. The indicator function takes the value one when the 
given condition is true and otherwise gives the value zero.
The next lemma gives the average cardinality of this code.

\begin{lemma}[Average cardinality of codes]
The average cardinality of $\Lambda(\theta(A), 0, \unlhd,  s)$ is given by 
\begin{equation} \label{avecard}
{\sf E}[M(A)] = n! \left(\frac{{n \choose 2}+{n^2-n \choose 2} }{{n^2 \choose 2} }   \right)^m,
\end{equation}
where the operator ${\sf E}$ denotes the expectation defined on $(S, P)$.
\end{lemma}
\begin{proof}
From the definition of $M(A)$, the expectation of the cardinality $M(A)$ can be written as
\begin{eqnarray} \nonumber
{\sf E}[M(A)] &=& \sum_{A \in S} P(A) M(A) \\
&=& \sum_{A \in S} P(A)\sum_{X \in \Pi_n} \Bbb {I}[\theta(A)\  \mbox{\sf vec}(X) \unlhd 0] .
\end{eqnarray}
By changing the order of summation, we can further transform this into 
\begin{eqnarray} \nonumber
{\sf E}[M(A)] 
&=& \sum_{X \in \Pi_n}\sum_{A \in S} P(A) \Bbb {I}[\theta(A)\  \mbox{\sf vec}(X) \unlhd 0] \\ \label{lem1}
&=& \frac{n!}{{n^2  \choose r}^m} \sum_{A \in S}  \Bbb {I}[\theta(A)\  \mbox{\sf vec}(X') \unlhd 0],
\end{eqnarray}
where $X'$ is an arbitrary permutation matrix in $\Pi_n$.  The last equality is due to 
the symmetry of the ensemble. Namely, this means that the quantity 
$
 \sum_{A \in S}  \Bbb{I}[\theta(A)\  \mbox{\sf vec}(X') \unlhd 0]
$
does not depend on the choice of $X'$. The evaluation of 
$
 \sum_{A \in S}  \Bbb{I}[\theta(A)\  \mbox{\sf vec}(X') \unlhd  0]
$
can be performed on the basis of the following combinatorial argument.

It is evident that any $X' \in \Pi_n$ contains $n$-ones as its components.
This implies that $ x' \defeq \sfvec(X')$ is a binary vector of length $n^2$ with Hamming weight $n$.
Let  $I_1 \defeq \{i \in [1, n^2] \mid x'_i = 1\}$, where $x_i' $ is the $i$th element 
of $ x'$.
Consider the first row of $A$, which is denoted by $a^T$. The relation 
$
 \theta(a^T)  x' = 0
$
holds if and only if 
\begin{eqnarray}
 |\{i \in I_1\mid a_i = 1\}| = 2\ \mbox{or } \  |\{i \in [1, n^2] \backslash I_1  \mid a_i = 1\}| = 2.
\end{eqnarray}
The number of possible ways to choose such a vector $ a$ is given by 
\begin{equation}
{n \choose 2}+{n^2-n \choose 2}.
\end{equation}
The term ${n \choose 2}$ corresponds to the number of possible ways such 
that $I_1$ (of cardinality $n$)  contains $2$-ones. On the other hand, 
${n^2-n \choose 2}$ represents the number of possible ways that remaining parts contains $2$-ones.
Since each row of $A$ can be chosen independently,  we consequently have 
\begin{equation} \label{Anum}
 \sum_{A \in S}  I[\theta(A)\  \mbox{\sf vec}(X') \unlhd 0] = \left({n \choose 2}+{n^2-n \choose 2} \right)^m.
\end{equation}
Substituting (\ref{Anum}) into (\ref{lem1}), we immediately obtain the claim of the lemma. 
\end{proof}

\begin{example}
In this experiment,
the number of $10 \times 10$ permutation matrices satisfying randomly generated equality 
constraints of two variables was counted. Figure \ref{averagesize} plots the cardinality of 
100-samples for the cases where  $m=30, 40, 50$.
The figure includes the ensemble average of the cardinality given by (\ref{avecard}) and 
the sample mean of the cardinality. The figure shows that cardinalities are scattered around
the ensemble average and that the sample mean agree with the ensemble average with reasonable 
accuracy.

This figure shows a trade-off relation between the number of additional equalities $m$ and
the cardinality. As (\ref{avecard})  indicates, the average cardinality is an exponentially 
decreasing function of $m$.
\hfill\fbox{}
\end{example}
\begin{figure}[htbp]
\begin{center}
\includegraphics[scale=0.9]{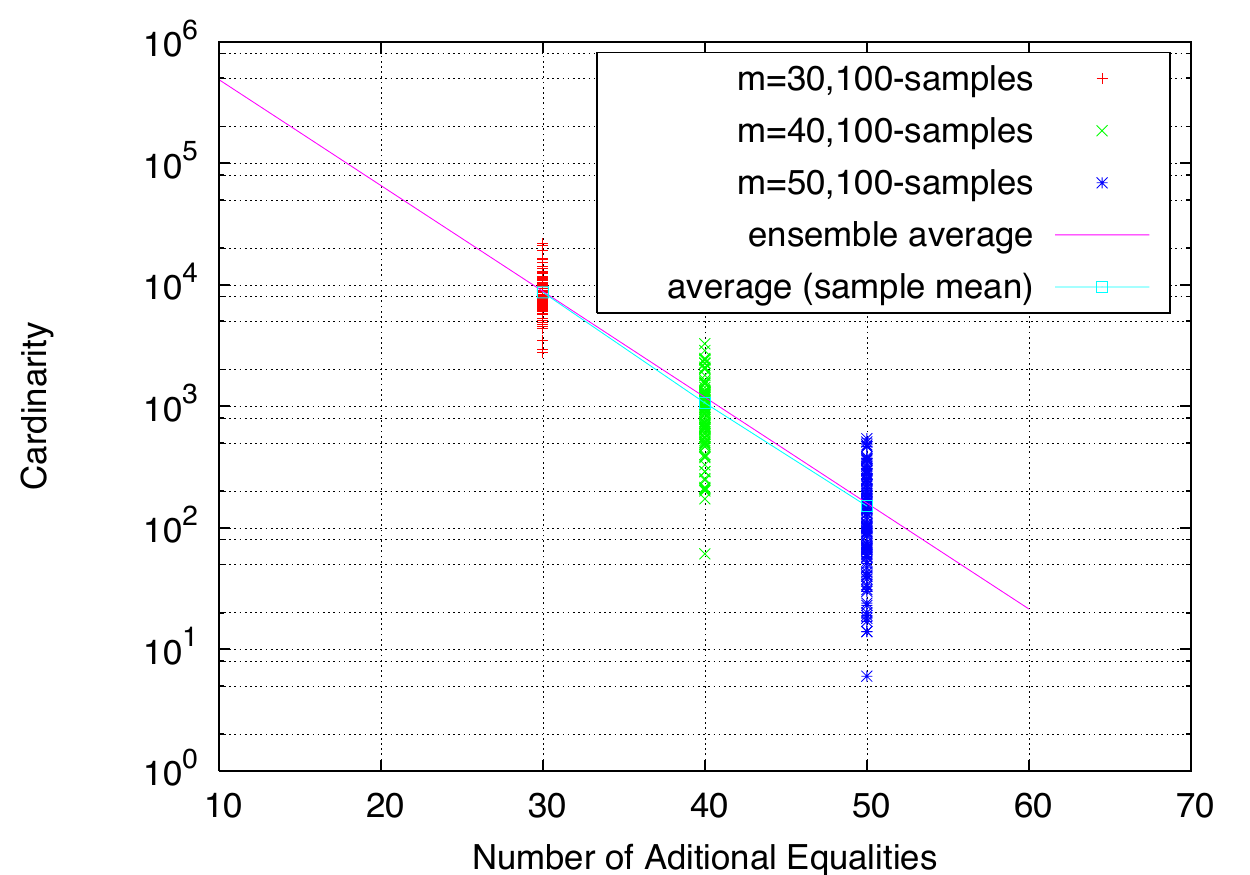} 
\caption{Relation between additional equalities $m$ and average cardinality}
  \label{averagesize}
\end{center}
\end{figure}

\subsection{Probabilistic analysis on weight distribution}

The origin $ o \defeq (o_1,\ldots, o_n)$ is an arbitrary permutation vector of length $n$; namely, 
$
 o \in \Lambda( s).
$
The number of codewords of $\Lambda(\theta(A), 0, \unlhd,  s)$
with Hamming weight $w$ is denoted by $L_w(A)$, where the 
Hamming weight $w_H(\cdot)$ is defined by
\begin{equation}
w_H( x) \defeq \sum_{i=1}^n \Bbb I[o_i \ne x_i],
\end{equation}
where $ x = (x_1,\ldots, x_n)$. This means the Hamming weight of $ x$ is equal to 
the Hamming distance between the origin and $ x$.
In other words, $L_w(A)$ is defined as 
\begin{equation}
L_w(A) \defeq \sum_{ x \in  \Lambda(\theta(A), 0, \unlhd,  s) } \Bbb I [w_H( x) = w].
\end{equation}
The set $\{L_1(A),\ldots, L_n(A) \}$ is referred to as the weight distribution of 
$\Lambda(\theta(A), 0, \unlhd,  s)$.

The next lemma gives the ensemble average of the weight distribution.
\begin{lemma}
The average weight distribution of the linearly constrained permutation code 
$\Lambda(\theta(A), 0, \unlhd,  s)$ is given by 
\begin{equation}
{\sf E}[L_w(A)] = {n \choose w} \left\lfloor \frac{w! + 1}{e} 
\right\rfloor \left( \frac{{n \choose 2} + {n^2-n \choose 2} }{{n^2 \choose 2}} \right)^m.
\end{equation}
\end{lemma}
\begin{proof}
The weight distribution $L_w(A)$ can also be expressed as 
\begin{equation}
L_w(A) = \sum_{X \in Z_w( o)} \Bbb I [\theta(A)\  \mbox{\sf vec}(X) \unlhd  0],
\end{equation}
where $Z_w( o)$ is defined by 
\begin{equation}
Z_w( o) \defeq \{X \in \Pi_n: w_H(X  s) = w\}.
\end{equation}

The expectation can be simplified as follows:
\begin{eqnarray}  \nonumber
{\sf E}[L_w(A)] 
&=& \sum_{A \in S} P(A)\sum_{X \in Z_w( o)} \Bbb I [\theta(A)\  \mbox{\sf vec}(X) \unlhd 0] \\ \nonumber
&=&\frac{1}{{n^2  \choose r}^m} \sum_{X \in Z_w( o)}  \sum_{A \in S} \Bbb I [\theta(A)\  \mbox{\sf vec}(X) \unlhd 0] \\ \nonumber
&=& \left( \frac{{n \choose 2} + {n^2-n \choose 2} }{{n^2 \choose 2}} \right)^m |Z_w( o) |. \\
\end{eqnarray}
The last equality is due to the symmetry of the ensemble and (\ref{Anum}).

The cardinality of $Z_w( o)$ is given by the following combinatorial argument.
Let $ x \in \Lambda( s)$ be an arbitrary vector satisfying $w_H( x)=w$.
The index set  $I_{diff}$ is defined by
$
I_{diff}( x) \defeq \{i \in [1,n] \mid o_i \ne x_i\}.
$
Let $T \subset [1,n]$ be an index set of cardinality $w$.
The quantity 
$
|\{ x \in \Lambda( s) \mid T = I_{diff}( x) \}| 
$
is equal to the number of derangements of length $w$, 
which  is known to be $\lfloor (w!+1)/e \rfloor$  \cite{IS}.
Note that the number of possible ways to choose $T$ is ${n \choose w}$.
Thus, we have the equality 
\begin{equation}
|Z_w( o) | = {n \choose w} \left\lfloor \frac{w! + 1}{e} \right\rfloor.
\end{equation}
This completes the proof of the lemma. 
\end{proof}
Note that the origin assumed here may not be included in $\Lambda(\theta(A),0,\unlhd,s)$.


\section{Conclusion}
\label{conclusion}
In this paper, a novel class of permutation codes, LP decodable permutation codes,  is introduced.
The LP decodable property is the main feature of this class of permutation codes.

The set of doubly stochastic matrices, i.e., the Birkhoff polytope, have 
$n!$ integral vertices which are permutation matrices. Additional linear constraints 
defines a code polytope which plays a fundamental role in the coding scheme presented in this paper.
An LP decodable permutation code is the set of integral vertices of a code polytope.

In an LP decoding process, a certain linear objective function is maximized under the assumption that 
the feasible set is a code polytope.
The decoding performance can be evaluated from geometrical properties of a code polytope.

The choice of additional linear constraints are crucial to construct  good codes.
In this paper, two approaches are discussed; namely, structured  permutation matrices and randomly constrained permutation matrices. 

Section \ref{structured} introduces some classes of structured linearly permutation matrices. Especially, it has been shown that
the pure involution codes have several nice properties; they are easy to encode and their error correction performance is much better than 
the trivial repetition code.

The random constraints discussed in Section \ref{random} enable us to use probabilistic methods for analyzing  
some properties of codes. The probabilistic methods \cite{Alon} are very powerful tool for grasping the relation between 
the number of constraints and important code parameters such as the cardinality of a code.

Although the paper provides fundamental aspects of the LP decodable permutation codes,
a number of problems remain still open. 
The following list is a part of open problems. 
\begin{enumerate}
\item Construction of good block permutation codes including a choice of an initial vector
\item Efficient algorithm for solving the LP problem arising in the LP decoding.
\item Permutation modulation for linear vector channels; let $H$ be a $n \times n$ real matrix. 
An ML decoding problem for a linear vector channel can be formulated as 
\begin{equation} \label{mdrule2}
\mbox{minimize } ||  y - H  x||^2 \mbox{ subject to }  x \in \Lambda(A,  b,\unlhd,  s).
\end{equation}
As discussed in this paper, the decoding problem can be relaxed to a quadratic programming (QP) problem:
\begin{equation} \label{mdrule2}
\mbox{minimize } ||  y - H  x||^2 \mbox{ subject to }  x \in {\cal P}(A,  b,\unlhd,  s).
\end{equation}
A QP-based decoding algorithm like \cite{wadayama2} appears interesting for this problem.
\item An application to rank modulation 
\end{enumerate}
Further investigation on related topics may open an 
interesting interdisciplinary research field among coding and  combinatorial optimization.

\section*{Appendix}

\subsubsection{Code polytopes for some classes of linearly constrained permutation matrices}

Table \ref{tightness} presents linear constraints for some sets of permutation matrices and 
their integrality of corresponding code polytopes. In this table, it is assumed that $X \in \Bbb R^{4 \times 4}$.
The integrality is numerically checked with the vertex enumeration program {\sf cdd} based 
on double description method by K. Fukuda  \cite{fukuda}.

\begin{table}[htdp]
\caption{Code polytopes and its properties ($n=4$)}
\label{tightness}
\begin{center}
\begin{tabular}{llll}
\hline
\hline
set of perm. matrices  & additional constraints & integrality & $|V|$\\
\hline
cyclic perm. mat. & (\ref{cyclic}) & Y & 4 \\
derangement  & ${\sf trace} (X)=0$ & Y & 9\\
involution & $X = X^T$ & N  & 14\\
transposition (1) & ${\sf trace} (X)=n-2$ & N & 20 \\
transposition (2) & ${\sf trace} (X)=n-2 $ & Y & 6 \\
 & $X=X^T$ &  & \\
$2\times 2$ block & constraints (\ref{subblock1}) & N & 28  \\
$2\times 2$ block & constraints (\ref{subblock1})  and (\ref{subblock2}) & Y & 8 \\
\hline
\end{tabular}
\end{center}
The column of integrality (Y/N) represents the code polytope is integral (Y) or not (N).
The column $\# V$ denotes the number of vertices on the code polytope.
\label{default}
\end{table}%

Some remarks on Table \ref{tightness} are listed as follows.
\begin{enumerate}
\item Cyclic permutation matrices
The cyclic permutation matrices of order 4 is given by the following additional linear constraints:
\begin{eqnarray} \nonumber
X_{1,1} = X_{2,2},\  X_{2,2}=X_{3,3}, \ X_{3,3}= X_{4,4} \\ \nonumber
X_{2,1}= X_{3,2},\ X_{3,2} = X_{4,3},\ X_{4,3} = X_{1,4} \\ \nonumber
X_{3,1}=X_{4,2},\ X_{4,2} = X_{1,3},\ X_{1,3}=X_{2,4} \\ \label{cyclic}
X_{4,1}= X_{1,2},\ X_{1,2}=X_{2,3},\ X_{2,3} = X_{3,4}.
\end{eqnarray}
In a similar way as in the case $n=4$, we can define the cyclic permutation matrices of order $n$.
The general expression the constraint for $n \times n$ cyclic permutation matrices is given by
\begin{equation}
\forall i, j \in [1,n],\quad X_{i,j} = X_{(i \mod n)+1, (j \mod n)+1}.
\end{equation}
\item Transposition: 
The permutation matrices satisfying the linear constraint ${\sf trace}(X)=n-2$ exactly coincides with 
the set of transpositions (i.e., permutations of two elements). Note that the constraint ${\sf trace}(X)=n-2$
does not give the tight polytope. Combining a redundant constraint $X=X^T$ (i.e., the involution constraint) to 
the trace constraint, the relaxed polytope becomes tight. This example indicates that redundant constraints
are necessary for constructing a tight polytope in some cases.  
\item Block constraint: The linear constraints for block permutation matrices (\ref{subblock1}) introduced in Theorem \ref{characterization}
does not give the tight polytope in $n=4$.
However, combining (\ref{subblock1}) and a set of redundant constraints (i.e., 90 degree rotation of (\ref{subblock1}))
\begin{eqnarray} \nonumber
X_{1,1} + X_{1,2} + X_{2,3} + X_{2,4}  &=& 1 \\ \nonumber
X_{2,1} + X_{2,2} + X_{1,3} + X_{1,4}  &=& 1 \\ \nonumber
X_{3,1} + X_{3,2} + X_{4,3} + X_{4,4}  &=& 1 \\  \label{subblock2}
X_{4,1} + X_{4,2} + X_{3,3} + X_{3,4}  &=& 1, 
\end{eqnarray}
we have the convex hull of $2 \times 2$ block permutation matrices. This case also shows importance of
redundant constraints from the optimization perspective. From this result,  it is expected that
the LP decoding performance of block permutation codes might be improved by incorporating these redundant 
linear equalities.
\end{enumerate}

\subsection*{Proof of Theorem \ref{characterization}}
\begin{proof}
In the first part of the proof, we will show that
any block permutation matrix satisfies (\ref{hagiwara}).

Assume that $k,b \in [1, \gamma]$ and $l \in [1, \nu]$ are arbitrary chosen.
From the definition of the skewed column set $U_{k,b}^{(l)}$, the left-hand side of (\ref{hagiwara})
can be rewritten as 
\begin{eqnarray} \nonumber
\sum_{ (u,v) \in U_{k, b}^{(l)} } X_{u,v} \hspace{-2mm}
&=& \hspace{-3mm}\sum_{ (u, v) \in T_{k, b}^{(l)} } X_{u,v} \\ \nonumber
&+& \hspace{-5mm}\sum_{k' \in [1, \gamma] \backslash \{k\}  } 
\left( \sum_{ (u, v) \in T_{k' , b}^{ (l \mod \nu)+1 } }  X_{u,v}  \right).\\ \label{eq:lemma:proof01}
\end{eqnarray}
Recall that  $X$ is assumed to be a block permutation matrix.
This means that there exists a unique block index $\kappa \in [1, \gamma]$ satisfying $X(B_{\kappa, b}) \ne  0$
for given block index $b$,  and the sub-matrix $X(B_{\kappa, b})$ is a permutation matrix.
If $k = \kappa$ holds, then 
\begin{equation}
\sum_{ (u,v) \in U_{k,b}^{(l)} } X_{u,v} = \sum_{(u,v) \in T_{k, b}^{(l)} } X_{u,v} = 1
\end{equation}
holds. Otherwise (i.e., $k \ne \kappa$), the equality 
\begin{equation}
 \sum_{ (u,v) \in U_{k,b}^{(l)} } X_{u,v} = \sum_{(u,v) \in T_{\kappa,b}^{  (l\mod \nu)+1 }} X_{u,v}   = 1.
\end{equation}
holds.  Thus, it has been proved that (\ref{hagiwara}) holds if $X$ is a block permutation matrix. 

We then move to the opposite direction; i.e., (\ref{hagiwara}) implies that $X$ is a block permutation matrix.

Assume that a block index $b \in [1, \gamma]$ and  a subindex $l \in [1, \nu]$ are arbitrary chosen.
Let $j = \nu (b-1) + l$. Since $X$ is a permutation matrix,
there exists the unique row index $i \in [1, n]$ satisfying $X_{i,j} = 1$.
The block $B_{k, b}$ containing the set of indices $(i,j)$ is uniquely determined 
because the blocks are mutually disjoint.
Under this setting, it is clear that $X(B_{k, b}) \neq 0$ holds.

In the following, we will show that
\begin{equation}
 k' \neq k \Rightarrow X(B_{k', b}) = 0.
 \end{equation}
From the definition of the block index $k$, It is clear that 
\begin{equation}\label{eq:lemma:proof02}
\sum_{(u,v) \in T_{k,b}^{(l)} } X_{u,v} = 1 
\end{equation}
holds.
Combining Eq. (\ref{eq:lemma:proof01}) and Eq. (\ref{eq:lemma:proof02}), we immediately obtain 
\begin{equation}\label{eq:lemma:proof03}
\sum_{k' \in [1, \gamma], k' \neq k  } \left( \sum_{ (u, v) \in T_{k' , b}^{ ( l\mod \nu )+1 } }  X_{u,v}  \right) = 0.
\end{equation}
This equality implies that 
\begin{equation}
(u,v) \in \bigcup_{k' \in [1, \gamma] \backslash \{k\}}  T_{k' , b}^{ ( l\mod \nu )+1 } \Rightarrow X_{u,v} = 0.
\end{equation}
Because $X$ is a permutation matrix, 
\begin{equation}
\sum_{(i,j) \in T_{k,b}^{( l \mod \nu )+1} } = 1 
\end{equation}
should be satisfied.
Applying the same argument iteratively, we consequently have
\begin{equation}
(u,v) \in \bigcup_{k' \in [1, \gamma] \backslash \{k\} } \bigcup_{l' \in [1, \nu]}  T_{k' , b}^{ ( l' ) } \Rightarrow X_{u,v} = 0.
\end{equation}
This statement is equivalent to  $ k' \neq k \Rightarrow X(B_{k', b}) = 0$. 
Due to the definition of the block permutation matrix, 
it has been proved that $X$ should be a block permutation matrix.
\end{proof}

\subsection*{Acknowledgement}

The first author appreciates Prof. Han Vinck for directing the author's attention to the field of 
permutation codes. The first author also wishes to thank Dr. Jun Muramatsu for  inspiring discussions on permutations.
This work was partly supported by the Ministry of Education, Science, Sports and Culture, Japan, Grant-in-Aid: 22560370.
We thank IBM academic initiative for IBM ILOG CPLEX Optimization Studio.

\subsection*{Biography}
\noindent
Tadashi Wadayama was born in Kyoto, Japan,on May 9,1968. 
He received the B.E., the M.E., and the D.E. degrees from Kyoto Institute of Technology 
in 1991, 1993 and 1997, respectively. Since 1995, he has been with Okayama Prefectural University 
as a research associate. In 2004, he moved to Nagoya Institute of Technology as an associate professor. 
Since 2010, he has been a professor of Department of Computer Science, 
Nagoya Institute of Technology.   His research interests are in coding theory, information theory, 
and digital communication/storage systems. He is a member of IEICE, and IEEE. \\[0.5cm]

\noindent
Manabu Hagiwara
received the B.E. degree in mathematics from
Chiba Univ. in 1997, and the M.E., and Ph.D.
degrees in mathematical science from the Univ.
of Tokyo in 1999 and 2002, respectively.
From 2002 to 2005 he was a postdoctoral fellow
at IIS, the Univ. of Tokyo. He also was a
researcher at RIMS, Kyoto University, 2002.
Currently, he is a research scientist of
Research Center for Information Security,
National Institute of Advanced Industrial Science
and Technology, and is an associated
professor of Center for Research and Development
Initiative, Chuo Univ. He also is a research
scholar at Univ. of Hawaii. His current research
interests include coding theory, cryptography,
information security, and algebraic combinatorics.

\end{document}